\documentclass[11pt]{article} % use larger type; default would be 10pt

\usepackage[utf8]{inputenc} % set input encoding (not needed with XeLaTeX)

\usepackage{geometry} % to change the page dimensions
\usepackage{graphicx} % support the \includegraphics command and options
\usepackage{color}

\definecolor{Red}{rgb}{1,0,0}
\definecolor{Blue}{rgb}{0,0,1}
\definecolor{Olive}{rgb}{0.41,0.55,0.13}
\definecolor{Green}{rgb}{0,1,0}
\definecolor{MGreen}{rgb}{0,0.8,0}
\definecolor{DGreen}{rgb}{0,0.55,0}
\definecolor{Yellow}{rgb}{1,1,0}
\definecolor{Cyan}{rgb}{0,1,1}
\definecolor{Magenta}{rgb}{1,0,1}
\definecolor{Orange}{rgb}{1,.5,0}
\definecolor{Violet}{rgb}{.5,0,.5}
\definecolor{Purple}{rgb}{.75,0,.25}
\definecolor{Brown}{rgb}{.75,.5,.25}
\definecolor{Grey}{rgb}{.5,.5,.5}

\usepackage{booktabs} % for much better looking tables
\usepackage{array} % for better arrays (eg matrices) in maths
\usepackage{paralist} % very flexible & customisable lists (eg. enumerate/itemize, etc.)
\usepackage{verbatim} % adds environment for commenting out blocks of text & for better verbatim
\usepackage{subfigure} % make it possible to include more than one captioned figure/table in a single float
\usepackage{amsmath,amssymb,amsthm,mathrsfs}
\usepackage[colorlinks]{hyperref}
\usepackage{blkarray}
\usepackage{bbm}

\usepackage{fancyhdr} % This should be set AFTER setting up the page geometry
\theoremstyle{plain}

\newtheorem{theorem}{Theorem}
\newtheorem{proposition}{Proposition}

\newtheorem{conj}{Conjecture}

\newtheorem{theorem*}{Theorem}   % no section numbers
\newtheorem{lemma*}{Lemma} % no section numbers
\newtheorem{corollary*}{Corollary} % no section numbers
\newtheorem*{remark*}{Remark}

\newlength{\widebarargwidth}
\newlength{\widebarargheight}
\newlength{\widebarargdepth}

\theoremstyle{definition}

\def\cD{{\cal D}}
\def\cE{{\cal E}}

\def\cQ{{\cal Q}}
\def\cR{{\cal R}}

\def\thetahat{\widehat{\theta}}
\def\Var{\operatorname{Var}}

\def\mprob{\mathbb P}

\newcommand{\real}{\ensuremath{\mathbb{R}}}
\newcommand{\defn}{\ensuremath{:  =}}
\newcommand{\E}{\ensuremath{\mathbb{E}}}
\renewcommand{\P}{\ensuremath{\mathbb{P}}}
\begin{document}

\begin{center}

{\bf{\Large{Teaching and learning in uncertainty}}}

\vspace*{.25in}

\begin{tabular}{ccc}
{\large{Varun Jog}} & \hspace*{.2in} & {\large{Po-Ling Loh}} \\
{\large{\texttt{vjog@wisc.edu}}} & & {\large{\texttt{ploh@stat.wisc.edu}}} \\
{\large{Department of ECE}} & & {\large{Department of Statistics}} \\
{\large{University of Wisconsin - Madison}} & & {\large{UW-Madison \& Columbia University}}
\end{tabular}

%\vspace{.2in}

%Departments of Electrical \& Computer Engineering$^\dagger$ and Statistics$^\ddagger$ \\
%University of Wisconsin - Madison \\
%Madison, WI 53706 \\

%\begin{tabular}{ccc}
%{\large{Varun \& Po-Ling Log}} \\
%{\large{\texttt{vjog@ece.wisc.edu}}} \vspace{.2in}
%\end{tabular}

%Department of Electrical and Computer Engineering \\
%Grainger Institute for Engineering \\
%University of Wisconsin - Madison \\ Madison, WI 53706

\vspace*{.2in}

October 2020

\vspace*{.2in}

\end{center}

\begin{abstract}
We investigate a simple model for social learning with two agents: a teacher and a student. The teacher's goal is to teach the student the state of the world; however, the teacher himself is not certain about the state of the world and needs to simultaneously learn this parameter and teach it to the student. We model the teacher's and student's uncertainties via noisy transmission channels, and employ two simple decoding strategies for the student. We focus on two teaching strategies: a ``low-effort" strategy of simply forwarding information, and a ``high-effort" strategy of communicating the teacher's current best estimate of the world at each time instant, based on his own cumulative learning. Using tools from large deviation theory, we calculate the exact learning rates for these strategies and demonstrate regimes where the low-effort strategy outperforms the high-effort strategy. Finally, we present a conjecture concerning the optimal learning rate for the student over all joint strategies between the student and the teacher.  %Our primary technical contribution is a detailed analysis of the large deviation properties of the sign of a transient Markov random walk on $\mathbb Z$. 
\end{abstract}

%\begin{IEEEkeywords}
%Social learning, large deviations, Markov chains
%\end{IEEEkeywords}

\section{Introduction}

Individuals in a society may learn about their environment directly through their own experiences, or vicariously via communication with other members of the society. Such interactions drive the exchange of ideas, technology, news, and opinions, and are critically important to the social and economic development of a society. However, understanding and predicting the effects of social interaction on society is a hard problem: each individual's opinion is dynamic and depends on his or her own biases, observations, and social interactions. The theoretical question of how agents learn through social interactions has consequently received much attention in the past few decades, and a number of mathematical models have been proposed to analyze social learning phenomena, such as those  detailed in Chamley~\cite{Cha04} and Mossel and Tamuz \cite{MosTam17}. 

Social learning models generally comprise an unknown state of the world and a number of agents. These agents have private observations of the state and use it to take actions to achieve a certain goal, such as maximizing their utility functions. Often, agents are able to observe the actions of some or all of the other agents, which they can use to glean more information about the state of the world and thereby play better actions. Broadly speaking, one is interested in analyzing the following questions in these models: (a) Convergence: Do the agents' actions eventually converge? (b) Agreement: Given convergence, do the agents agree? (c) Learning: Given agreement, is the unanimous opinion the true state of the world? and (d) Given learning, how fast does learning take place? 

In this paper, we focus on analyzing the \emph{rate of learning} posed in question (d) for a simple model with two agents. Our work is most closely related to the work of Vives~\cite{Viv93}, Jadbabaie et al.~\cite{MolEtAl13, MolEtAl17}, and Harel et al.~\cite{HarEtal14}, which also investigate learning rates in different social learning models. The specific learning model we consider is motivated by the work of Harel et al.~\cite{HarEtal14}, which we will describe in detail in Section~\ref{section: model} and contrast with our model.

A brief overview of our social learning model is as follows: The unknown state of the world $\Theta$ is drawn uniformly from $\{-1, +1\}$. The first agent, whom we call the teacher, receives repeated observations of $\Theta$ through a binary symmetric channel with flipping probability $p$. At each time instant, the teacher can transmit one bit over another binary symmetric channel with flipping probability $q$ to the second agent, whom we call the student. The teacher's goal is to ensure that the student learns the state of the world. Our goal is to understand how the rate of learning of the student depends on the joint strategies of the teacher and the student, where we assume that both the teacher and student are aware of each other's strategies. We analyze two particular student strategies: (i) A strategy where the student simply takes an average (or majority vote, when the state of the world is binary) over all observations received from the teacher; and (ii) a strategy where the student---shrewdly cognizant of the fact that the teacher is also learning from his own observations over time---only averages a final segment of observations, which she assumes to more accurately reflect the state of the world than the initial observations.

The main mathematical tools we use in this paper are derived from the theory of large deviations, specifically concerning the large deviation properties of the sign of a transient Markov random walk on $\mathbb Z$. We show that the rate function of this process can be explicitly calculated; moreover, it has a surprisingly neat closed-form expression that can be used in our subsequent analysis. When the state of the world is binary, our results show that when the student employs either of the two learning strategies discussed above, the relative ordering of teacher strategies in terms of the learning rate of the student generally depends on the noise parameters of the teacher's and student's channels, in a manner we can make explicit. We also analyze a setting where the state of the world is a continuous parameter and the learning rate is quantified by the variance of the student's estimator, and reach a markedly different conclusion that the laziest strategy of the teacher already leads to an optimal rate of learning for the student.

The remainder of the paper is organized as follows: In Section \ref{section: model}, we describe our model in detail and discuss the strategies employed by the teacher and student. In Section \ref{section: markov}, we develop the main technical tools used in analyzing the learning rate of the student when the state of the world is binary; the rates of learning for the two student strategies are subsequently analyzed in Sections~\ref{SecMajority} and~\ref{SecEpsMajority}. In Section~\ref{SecGaussian}, we discuss the somewhat different case of Gaussian learning with continuous-valued parameters. Finally, we conclude the paper in Section \ref{section: end} with a discussion of open problems and teaching philosophy.

\textbf{Notation:} For a real parameter $p \in [0,1]$, we write $\bar p$ to denote the quantity $1-p$. For a random variable $X_n$, a set $A$, and a function $f(A)$, we write $\mprob(X_n \in A) \approx e^{-n f(A)}$ to mean that $\mprob(X_n \in A) = e^{-n(f(A) + o(1))}$. We write $D(a || b) = a \log (a/b) + \bar a \log (\bar a/ \bar b)$ to denote the Kullback-Leibler divergence between the Bernoulli($a$) and Bernoulli($b$) distributions. We write $\textbf{1} \in \real^n$ to denote the $n$-dimensional vector of all 1's.

%%%%%

\section{Problem statement and related work}
\label{section: model}

We consider a simple model of social learning with two agents: a teacher and a student. Suppose both agents are trying to learn an unknown binary random variable $\Theta$, which is called the state of the world. We assume that $\Theta$ takes values in the set $\{-1, +1\}$, uniformly at random. At each time $i \geq 1$, the teacher observes a noisy version of $\Theta$ through a binary symmetric channel with parameter $p \in [0, 1/2)$; i.e.,
\begin{align*}
\P(Y_i = \Theta) = 1-p, \quad \text{ and } \quad \P(Y_i = -\Theta) = p.
\end{align*}
Conditioned on $\Theta$, the random observations $\{Y_i\}_{i \geq 1}$ are independent and identically distributed, as above. The student does not make any direct observations (noisy or otherwise) of $\Theta$, and may only learn its value from the teacher.

At each time $i$, the teacher communicates a binary random variable $\hat X_i$, which is a (possibly random) function of the history of observations $\{Y_j\}_{1 \leq j \leq i}$, and the student receives a noisy version of $\hat X_i$, which we call $Z_i$. The communication channel from the teacher to the student is assumed to be a binary symmetric channel with parameter $q \in [0, 1/2)$. The student's estimate of $\Theta$ after observing $\{Z_j\}_{j \leq i}$ is denoted by $\hat \Theta_i \in \{-1, +1\}$. We refer to the sequence of random variables $\{\hat X_i\}$ as the teacher's strategy, and the decoding rules $\{\hat \Theta_i\}$ as the student's learning strategy. For fixed teaching and learning strategies, the student's rate of learning is defined as follows:
\begin{align}\label{eq: learning rate}
\cR = \limsup_{n \to \infty} \left\{-\frac{1}{n} \log \P\left( \hat \Theta_n \neq \Theta \right)\right\}.
\end{align}
Notice that the teacher is guaranteed to the learn the state of the world eventually, owing to his repeated noisy observations of $\Theta$.

\paragraph{Comparison to Harel et al.~\cite{HarEtal14}:} Despite the differences between our model and that of Harel et al.~\cite{HarEtal14}, we also uncover certain counterintuitive phenomena in our two-agent setting. The model in Harel et al.~\cite{HarEtal14}\footnote{See version 1 of the arXiv manuscript. In later versions, the model was extended to more than two agents, but the key ingredients of the analyses may be found in the analysis of the two agent model.} also considers two agents $A$ and $B$ and an unknown binary state of the world $\Theta$. The agents receive i.i.d.\ observations $\{A_i\}_{i \ge 1}$ and $\{B_i\}_{i \ge 1}$ of $\Theta$ through their respective binary symmetric channels. At each time $i$, the agents also form their binary-valued best estimates $\widehat \theta_A$ and $\widehat \theta_B$ of $\Theta$. The information available to the agents is modeled in two settings: (i) $A$ can observe $B$'s estimates of $\Theta$, but not vice versa; and (ii) $A$ and $B$ can both observe each other's estimates. The authors made the surprising observation that agent $A$'s learning rate is higher in setting (i) than in (ii)---contrary to the intuition that setting (ii) involves a greater exchange of information, so one might expect $A$ to learn faster. The authors attribute this counterintuitive result to a phenomenon they call ``rational groupthink.'' 

The model studied in our paper differs from that of Harel et al.~\cite{HarEtal14} in three key ways: First, the second agent, the student, does not have any private observations that allow her to learn. Any information she receives about the state of the world follows from a noisy interaction with the teacher. Second, the agents are not necessarily Bayesian, but instead perform heuristic calculations to form their opinions. Rich bodies of literature studying both Bayesian and non-Bayesian models of social learning exist, which we shall describe in more detail in Section~\ref{section: signal}. Third, the model proposed in Harel et al.~\cite{HarEtal14} involves \emph{pure information externalities}, where each agent receives a payoff which depends only on their action and the state of the world. In contrast, our model does not include payoff functions for agents; rather, the teacher and student are jointly working to optimize the student's learning rate. Similar features appear in team decision theory and control, which we briefly comment on in relation to our setting in Section~\ref{section: control}.

Despite the differences between our model and that of Harel et al.~\cite{HarEtal14}, we also uncover certain counterintuitive phenomena in our two-agent setting. In particular, we observe that ``helpful" social interactions, where the teacher always tries to transmit his best guess to the student, may actually hinder the student's rate of learning. \\

In Section~\ref{SecGaussian} below, we will introduce an alternative model for teacher/student learning when the state of the world is a continuous real number. The rate of learning of the student will then be quantified by the variance of the student's estimate, rather than the probability of error. We will introduce the appropriate terminology later; in the remainder of this section and throughout Sections~\ref{SecMajority} and~\ref{SecEpsMajority}, we will adopt the setting and notations for the binary model described above.

%%%%%

\subsection{Student strategies}

The student's learning rate depends on both the teacher's strategy and her own decoding strategy. When the student is aware of the teacher's strategy, the optimal learning rate is achieved when the student uses a maximum likelihood decoder to arrive at her estimate of $\hat \Theta_n$. However, as is well-documented in the literature on social learning, a fully rational model often places unreasonable computational demands on Bayesian agents \cite{MosTam17}. Assuming that agents are non-Bayesian serves two goals: it makes the model more realistic by reducing its complexity; and in some cases, it also helps make the model mathematically tractable. In this paper, we will consider two simple non-Bayesian student strategies, which we now describe.

\textbf{Majority rule:} This is perhaps the simplest possible strategy for the student, defined as a majority vote over her observations:
\begin{align*}
\hat \Theta_n = 
\begin{cases}
+1 &\text{ if } \quad \frac{\sum_{i=1}^n \mathbbm{1}(Z_i = +1)}{n} \geq \frac{1}{2},\\
-1 &\text{ otherwise.}
\end{cases}
\end{align*}
Note that in some cases (e.g., when the teacher uses the simple forwarding strategy defined in the next section), the majority strategy for the student corresponds to the MLE; however, this is not generally true when the teacher employs other strategies.

\textbf{$\epsilon$-majority rule:} This is a generalization of the majority learning rule, where the student's estimate is a majority vote among the latter $\epsilon n$ observations, for a parameter $\epsilon \in (0, 1]$:
\begin{align*}
\hat \Theta_n = 
\begin{cases}
+1 &\text{ if } \quad \frac{\sum_{i=\lceil(1-\epsilon)n\rceil}^n \mathbbm{1}(Z_i = +1)}{\lfloor \epsilon n\rfloor} \geq \frac{1}{2},\\
-1 &\text{ otherwise.}
\end{cases}
\end{align*}
The rationale is that the student is aware that the teacher is learning as time progresses, so she may be skeptical of the teacher's initial transmissions and prefer to place more weight on the most recent observations. However, since analyzing the learning rate of the student becomes rather complicated for arbitrary weighting strategies, we will restrict our analysis to a strategy that places zero weight on the first $(1-\epsilon)n$ observations and equally weights the remaining observations. \\

Note that the majority rule is a special case of the $\epsilon$-majority strategy when $\epsilon = 1$. However, as our analysis will reveal, the optimal choice of $\epsilon$ depends in a nontrivial manner on the channel parameters $p$ and $q$. Thus, we will generally be interested in the behavior of $\epsilon$-majority learning when $\epsilon$ is optimized to the parameters $p$ and $q$, operating under the assumption that such a strategy would only be employed if the student had some knowledge of the channel parameters. We will also analyze the majority learning strategy on its own.

%%%%%

\subsection{Teacher strategies}

%In a majority learning model, the student will learn what they hear most often, and therefore the teacher should try to teach the ``correct lesson" more often than the ``wrong lesson." What are be some natural strategies that the teacher might employ?

We now turn to describing several teaching strategies that we will analyze in our paper.

%{\red
%Furthermore, in case (ii), we assume that the student is aware of the noise parameters governing the conditional distributions of both the teacher's and student's observations, and makes the optimal choice of $\epsilon$ with respect to these parameters. We then characterize the overall learning rate of the student when the teacher employs different strategies: a ``lazy" strategy, where he simply forwards his observations of the state of the world directly to the student; a cumulative teaching strategy, where he sends his best guess of the state of the world to the student at each time step; and a hybrid strategy, specifically designed for learning strategy (ii), in which the teacher learns during the first $(1-\epsilon)n$ time steps and then repeatedly transmits his best guess of the state of the world in the remaining $\epsilon n$ time steps. Although the third teaching strategy should not be expected to produce an optimal learning rate---since the teacher effectively stops learning after the initial $(1-\epsilon)n$ observations---it nonetheless leads to a convenient closed-form expression for the student's learning rate, and is a reasonable mathematical model for a teaching strategy commonly encountered in practice.
%}

\textbf{Simple forwarding:} A ``lazy" strategy for the teacher that requires no learning on his part is to put $\hat X_i = Y_i$; i.e., simply forward his observation at each time step directly to the student.

\textbf{Cumulative teaching:} In contrast to the lazy strategy of simply forwarding information, the teacher might follow a strategy of always transmitting his current best estimate of $\Theta$, obtained by applying the majority rule to his observations $\{Y_i\}_{1 \leq i \leq n}$. %Analyzing the learning rate for this strategy is the main problem we tackle in this paper. \\

Note that the cumulative teaching strategy clearly satisfies the property that after some finite time, the process $\{\hat X_n\}$ is identically equal to $\Theta$. In contrast, the simple forwarding strategy never converges in this manner, so the teacher is correct more often in the cumulative teaching strategy if $n$ is sufficiently large. This is the intuitive reason why one might expect the cumulative teaching strategy to dominate the simple forwarding strategy; however, as we will see later, the relative merits of the two strategies are intricately linked to the values of $p$ and $q$. %In what follows, we calculate the exact learning rate for this strategy.

\textbf{$\epsilon$-teaching:} We will also analyze a teaching strategy that is tailored to the $\epsilon$-majority student strategy. In this strategy, the teacher transmits no information during the first $(1-\epsilon)n$ time steps (e.g., always transmitting a default value of  $1$), and then repeatedly transmits his best guess based on the first $(1-\epsilon)n$ observations in the remaining time steps:
\begin{equation*}
\hat{X}_i =
\begin{cases}
+1 & \text{ if } \quad \frac{\sum_{i=1}^{\lfloor (1-\epsilon) n\rfloor} \mathbbm{1}(Y_i = +1)}{\lfloor (1-\epsilon) n \rfloor} \ge \frac{1}{2}, \\
-1 & \text{ otherwise},
\end{cases}
\qquad \text{for } i \ge \lceil (1-\epsilon)n \rceil.
\end{equation*}
Evidently, the $\epsilon$-teaching strategy could potentially be improved if the teacher transmitted more information in the first $(1-\epsilon)n$ time steps, or transmitted a more sophisticated estimator based on cumulative learning in the last $\epsilon n$ time steps. However, we again focus on analyzing this simpler strategy since it leads to closed-form expressions for learning rates that may be more easily compared to other student/teacher strategies.

%%%%%
\begin{remark*}
The definition of learning rate in equation~\eqref{eq: learning rate} is motivated, at least in part, for reasons of mathematical tractability. It is natural to wonder how large $n$ should be so that the probability of error is close to the approximation resulting from an asymptotic calculation. Calculating the exact probability of error for large $n$ is not computationally feasible for $n > 20$; however, the error probability can be approximated using Monte Carlo simulations for moderate values of $n$. Based on our experiments for the ``cumulative teaching + majority learning'' strategy, we observe that although the values of $p$ and $q$ generally affect the outcome, in most cases, the asymptotic approximation becomes reasonably accurate when $n \approx 100$. Figure~\ref{fig: convergence} displays plots for the error probability vs.\ $n$ when $(p,q) = (0.1, 0.3)$ and $(p,q) = (0.3, 0.1)$.
\end{remark*}

\begin{figure}[h]
\begin{center}
\begin{tabular}{cc}
\includegraphics[width=3in]{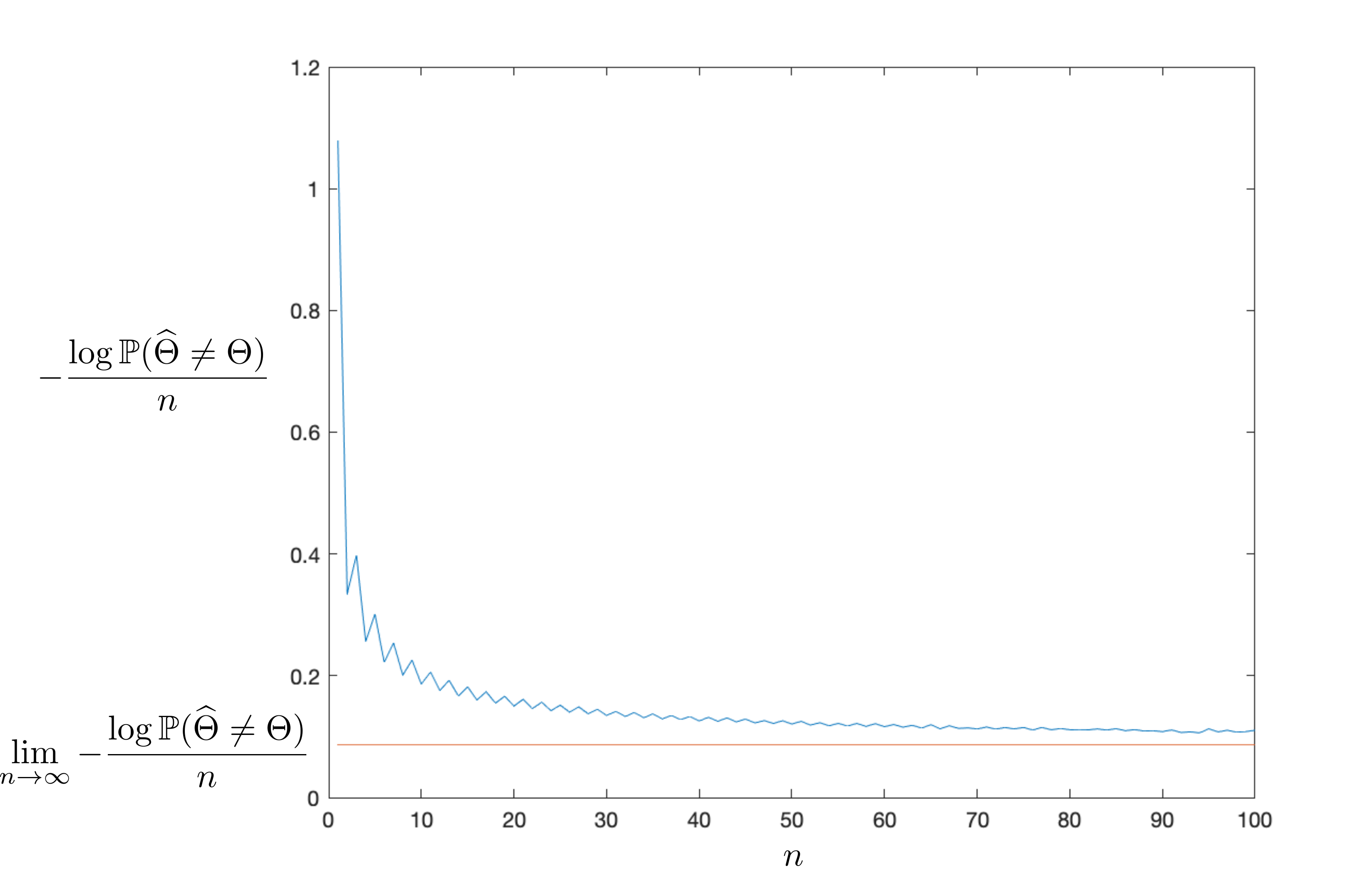} & \includegraphics[width=3in]{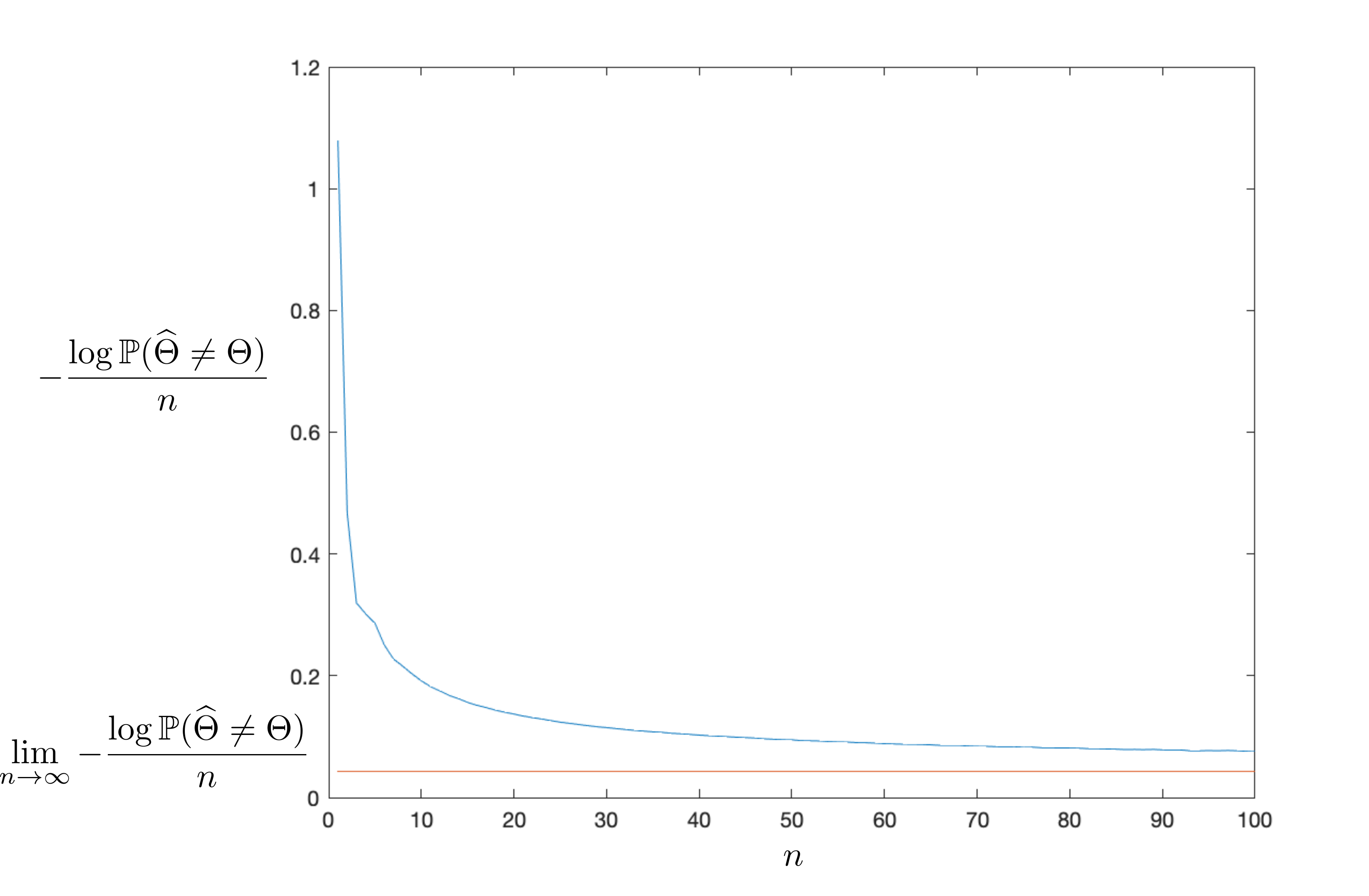} \\
(a) $p = 0.1$ and $q = 0.3$ & (b) $p = 0.3$ and $q = 0.1$
\end{tabular}
\caption{Plots showing the convergence of the non-asymptotic error probability of the student in the ``cumulative teaching + majority learning'' strategy to the exact learning rate defined by equation~\eqref{eq: learning rate}.} \label{fig: convergence}
\end{center}
\end{figure}

\subsection{Overview of results}

To aid readability, we now provide a preview of our main results. As mentioned earlier, we will focus on characterizing the learning rate of different strategies when the student's strategy is fixed. The main message of our paper is that certain teacher strategies are better than other strategies for particular regimes of the channel parameters $(p,q)$:
\begin{itemize}
\item When the student employs the majority rule strategy, neither the simple forwarding strategy nor the cumulative teaching strategy strictly dominates the other. In Section \ref{SecMajority}, we calculate the analytical expressions for the learning rates in both cases. Comparing them for different values of $p$ and $q$, we observe that for small values of $q$ (which one may interpret as sharp and attentive students), the simple forwarding strategy dominates. However, we also note that for all $q$ larger than $\approx 0.15$, the simple forwarding strategy is worse than the cumulative teaching strategy for all values of $p$; i.e., no matter how ``bad'' the teacher (corresponding to large values of $p$), a moderately attentive student still benefits from a cumulative teaching strategy. 
\item When the student employs the $\epsilon$-majority strategy, neither the simple forwarding strategy nor the cumulative teaching strategy strictly dominates the other. However, we can show that the cumulative teaching strategy uniformly dominates the $\epsilon$-teaching strategy for all values of $p$ and $q$. We can obtain closed-form expressions for the learning rate of the latter teaching strategy; for the former teaching strategy, we obtain an expression that can be computed by solving an optimization problem using computer software. Comparing the cumulative teaching strategy with the simple forwarding strategy, we observe that the cumulative teaching + $\epsilon$-majority learning strategy dominates the simple forwarding + majority learning strategy for almost all values of $p$ and $q$, except when $q$ is very small.
 
\end{itemize}
%When analyzing the $\epsilon$-majority student strategy, we assume that the student is aware of the channel parameters, as well as the strategy employed by the teacher, and optimizes $\epsilon$ to obtain the best learning rate.

In the case of Gaussian teaching and learning, we show that rather surprisingly, the majority learning rule for the student and simple forwarding strategy for the teacher is jointly dominant. Notably, neither of these teaching or learning strategies requires knowledge of the Gaussian channel parameters. This underscores a fundamental difference between the social learning problem in discrete vs.\ continuous state spaces.
%%%%%

\subsection{Related literature}
\label{section: related work}

The problem described in the section above possesses additional connections to various research threads spanning a diverse range of topics. We briefly highlight some of these below.

\subsubsection{Social learning and signaling games}
\label{section: signal}

The economics literature contains a vast body of work concerning social learning. We refer the reader to the survey articles by Mossel and Tamuz~\cite{MosTam17} and Mobius and Rosenblat~\cite{MobRos14} for a broad coverage, and only highlight a subset of this literature that is most relevant to our work.

Aumann's Agreement Theorem~\cite{Aum76} may be viewed as an early example of social learning. Aumann showed that if two agents agree on the prior of state of the world and their posteriors are ``common knowledge,'' their posteriors must be identical; i.e., they cannot agree to disagree. Geanakoplos and Polemarchakis~\cite{GeaPol82} studied the speed of convergence of the agents' posteriors when they exchange and update their posteriors at each time step. Social learning situations where agents update their actions repeatedly---based on their private signals and the (continuously updating) actions of other agents---have been studied in numerous popular models in economics. Banerjee~\cite{Ban92}, Bikhchandani et al.~\cite{BikEtal92}, and Smith and Sorenson~\cite{SmiSor00} proposed models of social learning which give rise to the phenomenon of herding, where agents ``herd'' to the same (possibly wrong) action.

Gale and Kariv~\cite{GalKar03} proposed a social network model, wherein agents occupy vertices of a graph and are able to observe the actions of their neighbors. Akin to Aumann's result~\cite{Aum76}, Gale and Kariv showed that fully Bayesian (or fully rational) agents converge to the same action under suitable conditions. Mossel et al.~\cite{MosEtal14} studied learning in this model by incorporating a state of the world that dictates private signals of the agents. The model studied in our paper may be considered as operating on a social network graph with one edge: the student can noisily observe the teacher's actions. Nonetheless, the information structure is somewhat different in our model, which is more closely related to the model proposed in Harel et al.~\cite{HarEtal14}.

Another notable aspect of our model is that the two agents are non-Bayesian. Learning models with non-Bayesian agents have been studied by various authors, e.g., DeGroot~\cite{Deg74}, Ellison and Fudenberg~\cite{EllFud93}, Bala and Goyal~\cite{BalGoy98}, Rahimian and Jadbabaie~\cite{RahJad16}, Haz{\l}a, Jadbabaie, Mossel, and Rahimian \cite{JadEtal17} and Molavi et al.~\cite{MolEtAl17}. Non-Bayesian models are relevant for two reasons: (1) humans are not Bayesian; and (2) it can be challenging to analyze the behavior of Bayesian agents. This is remarked upon in Gale and Kariv~\cite{GalKar03} in the study of a graph with just three nodes and also explored further in Kanoria and Tamuz~\cite{KanTam13}. Consequently, we have chosen to consider a model in which agents make heuristic rather than Bayesian decisions.

\subsubsection{Team decision theory and control}
\label{section: control}

Team decision theory, as described in Ho~\cite{Ho80}, considers agents with correlated sources of private information who take coordinated actions to optimize a payoff function. The agents may communicate beforehand to agree on any protocol of their choice. Similarly, our model involves a teacher and student who share a common goal and are free devise a joint strategy; however, the specifics of the mathematical model in Ho~\cite{Ho80} differ from ours in several ways, particularly with respect to the repeated observations and actions in our setting. 

Our model also shares several similarities with Witsenhausen's counterexample from stochastic control~\cite{Wit68}. This counterexample, which involves two agents---one with a perfect observation of the state but expensive control, and the other with a noisy observation of the state (after it has been modified by the first agent) but free control---demonstrates how non-classical information structures can alter the nature of optimal control solutions. The goal of both agents is to drive  the state to 0. The problem may be reduced to finding the optimal strategy of the first agent, since the second agent can use a simple MMSE decoder. In our model, the teacher has, in some sense, direct access to the state of the world; the teacher can then communicate this state to the student in a noisy manner. Just as in the optimal control problem, our goal is to determine the teacher's optimal strategy, since the student's optimal strategy is simply the MAP decoder once the teacher's strategy is fixed. (A notable difference is the absence of a ``control'' aspect in our setting.) Nonetheless, the difficulty in identifying the optimal strategies in Witsenhausen~\cite{Wit68} and our paper exemplifies how asymmetric information structures can lead to very hard problems that might initially appear innocuous.

\subsubsection{Fluctuation of random walks}

Analyzing the sign of the random walk appearing in our paper may be seen as a part of the broader literature concerning fluctuation theory in probability \cite{And74}. These works explore the distributions of various quantities such as the time to reach the maximum or minimum, and (of relevance to us) the sojourn time, which is the duration of time when the random walk is positive. Fluctuation theory is also closely related to ballot theorems in probability, where the probability of a random walk being positive, negative, or lying within certain bounds is studied~\cite{AddRee08}. Our contribution differs due to its focus on the large deviation properties of the sign, as opposed to characterizing its distribution as in Chung and Feller~\cite{ChuFel49}, Andersen~\cite{And55}, and several works along these lines. We note that the paper of Andersen~\cite{And55} contains a result concerning the sign of a random walk that may be used to prove our Theorem~\ref{thm: magic} more directly. However, we have retained our original analysis (which is a brute force computation, as opposed to the specialized result from Anderson~\cite{And55}), because the technique extends to non-binary random walks and the calculation of rate functions for the duration of time the random walk lies within a certain interval. A possible generalization of our model to continuous random walks may be possible to study using arcsine laws~\cite{Lev40, Dur19}, but we do not explore that in this paper.

\subsubsection{Communication theory}

A model identical to ours was proposed independently and concurrently in Huleihel et al.~\cite{HulEtal19}, who studied how to reliably send one bit across a cascade of binary symmetric channels. For simplicity, Huleihel et al.~\cite{HulEtal19} assume that all the binary symmetric channels in the cascade are identical, with flipping probability $\delta$. Just as in this paper, the authors examine the learning rate (which they call the ``information velocity") as a function of the number of channels $k$ in the cascade. The paper establishes that when $\delta \to 1/2$, the ratio of the learning rates for $k=2$ and $k=1$ is at least $3/4$, and proposes an interesting conjecture that the ratio should be arbitrarily close to $1$. Translated into our setting, this would mean that when the teacher and student both have very noisy channels, the student learns as fast as the teacher; i.e., the channel from the teacher to the student can be made ``clean.'' We shall further discuss the conjecture from Huleihel et al.~\cite{HulEtal19} in our concluding section along with other conjectures of interest.
%%%%%

\section{Technical results}
\label{section: markov}

The following results, related to the sign sequence of a random walk and its derived properties, will be used throughout our paper. They provide the main technical results underlying our calculations of explicit learning rates. Throughout our analysis, we assume WLOG that $\Theta = +1$.

%Let $\hat X_n = \sign\left(\sum_{i=1}^n Y_i\right)$ denote the value transmitted by the teacher at time $n$, which is his current best guess of the state of the world at time $n$.
The random process $X_n \defn \sum_{i=1}^n Y_i$, recording the cumulative observations of the teacher, may be modeled as a random walk $\mathbb Z$ with the following transition probabilities:
\begin{align*}
&p(X_{n+1} = i+1 | X_n = i) = 1-p, \qquad \text{and}\\
&p(X_{n+1} = i-1 | X_n = i) = p,
\end{align*}
where $p < 1/2$. Notice that this random walk is transient; i.e., for every $i \in \mathbb Z$, the random walk visits state $i$ finitely many times, with probability $1$. Since $p < 1/2$, the random walk eventually runs off to $+\infty$. Now define the following process:
\begin{align*}
\hat X_n = 
\begin{cases}
+1 &\text{ if } X_n > 0,\\
-1 &\text{ if } X_n < 0,\\
\hat X_{n-1} &\text{ if } X_n = 0,
\end{cases}
\end{align*}
which is the teacher's best guess about the state of the world at time $n$.

Let $M_n := {\sum_{i=1}^n \mathbbm{1} (\hat X_i = +1)}$ denote the number of times that the teacher's majority guess is correct up to time $n$. In order to determine learning rates for the cumulative teaching strategy, we will need to explore the large deviations behavior of $M_n$. In particular, we are interested in the quantity $\P \left(\frac{M_n}{n} \approx 1-\delta\right)$. We expect this probability to be approximately equal to $e^{-nf(\delta)},$ for some suitable exponent $f(\delta)$; in Section~\ref{SecMn}, we will pinpoint the function $f(\delta)$ in terms of $p$ and $\delta$. 

%%%%%

\subsection{Preliminary calculations for $\{X_n\}$}

Since the random walk $\{X_n\}$ is transient, it has a positive probability of never returning to state $i$, starting from state $i$. A simple calculation shows that this probability is $1-2p$, for any value of $i$.

Next, we focus on the sojourn time $T$, defined to be the time of the \emph{first} return to 0, starting from 0. We use the convention that $T$ is positive if the random walk is positive during the sojourn; otherwise, $T$ is negative. Note that sojourn times can only take even values: $T=2k$ when the random walk takes a total of $k$ positive steps and $k$ negative steps, with only the endpoints of the sojourn being at 0. The probability that $T = 2k$ may be calculated by counting the number of such paths and multiplying the result by $p^k\bar p^k$. Furthermore, the number of these paths is the $(k-1)^{\text{st}}$ Catalan number \cite{Sta15}, $C_{k-1} = \frac{1}{k} {2k-2 \choose k-1}.$ The distribution of $T$ is therefore given by
\begin{align*}
\mathbb P(T = 2k) = 
\begin{cases}
p^{|k|}\bar p^{|k|} \frac{1}{|k|} {2|k|-2 \choose |k|-1} &\text{ if } 0< |k| < \infty,\\
0 &\text{ if } k = 0,\\
1-2p &\text{ if } k = +\infty,\\
0 &\text{ if } k = -\infty.
\end{cases}
\end{align*}

We will be interested in the random variable $T$ \emph{conditioned on} the event that $|T| < \infty$, which we call $\tilde T$. It is easy to see that $\mathbb E \tilde T = 0$, and
\begin{align*}
\mathbb E |\tilde T| &= \sum_{k=1}^\infty 2\cdot \frac{p^k\bar p^k}{2p} C_{k-1} \cdot 2k
= 2\sum_{k=0}^\infty \frac{p^{k+1}\bar p^{k+1}}{p} C_{k} \cdot (k+1).
%&= 2\bar p\sum_{k=0}^\infty (p\bar p)^k C_{k} (k+1).
\end{align*}
Recall that the generating function of the Catalan numbers is given by
\begin{align}
\label{EqnGenFunc}
f(x) = \sum_{k=0}^\infty C_k x^k = \frac{1- \sqrt{1-4x}}{2x}.
\end{align}
Thus, 
$\sum_{k=0}^\infty C_k (k+1) x^k = \frac{d}{dx} (xf(x))= \frac{1}{\sqrt{1-4x}}.$
Substituting $x = p\bar p$, we conclude that $\mathbb E |\tilde T| = \frac{2\bar p}{(\bar p-p)}$.

Another quantity that will be critical for deriving large deviations results is the log moment generating function of the random vector $(\tilde T, |\tilde T|)$, defined by
\begin{equation*}
L(\lambda_1, \lambda_2) = \log \mathbb Ee^{\lambda_1 \tilde T + \lambda_2 |\tilde T|}.
\end{equation*}
The following lemma is proved in Appendix \ref{proof: lemma: gamma}:
\begin{lemma*}
\label{lemma: gamma}
Let
\begin{equation*}
\cD = \{(x_1, x_2) \mid |x_1| + x_2 \le D(1/2 || p) \}.
\end{equation*}
For $(\lambda_1, \lambda_2) \in \cD$, we have $L(\lambda_1, \lambda_2) \leq \log \frac{1}{2p}$. For $(\lambda_1, \lambda_2) \notin \cD$, we have $L(\lambda_1, \lambda_2) = +\infty$.
\end{lemma*}

Finally, the last ingredient we need is the number of returns of the random walk $\{X_n\}$ to 0. This is a geometric random variable, with distribution given by
\begin{align*}
\mathbb P(G = i) = (2p)^{i}(1-2p), \quad i \geq 0.
\end{align*}

%%%%%

\subsection{Large deviations properties of $M_n$}
\label{SecMn}

Let $B$ be the random variable indicating the time of the final visit of $\{X_n\}$ to state $0$. We break up the probability of interest as follows:
\begin{equation}
\label{EqnMn}
\mathbb P\left(\frac{M_n}{n} \leq 1-\delta\right)%&= \mathbb P(M_n \leq n(1-\delta), B \leq n) + \mathbb P(M_n \leq n(1-\delta), B > n)\\
= \left(\sum_{g = 0}^{n/2} \mathbb P(M_n \leq n(1-\delta), B \leq n, G = g)\right) + \mathbb P(M_n \leq n(1-\delta), B>n).
\end{equation}
Note that the sum contains fewer than $n$ terms (the maximum number of returns to 0 in $n$ steps is $\frac{n}{2}$), and the largest among these terms will dictate the exponential growth rate of the sum.

We have the following lemma, proved in Appendix~\ref{AppLemMnSecond}:
\begin{lemma*}
\label{LemMnSecond}
We have
\begin{equation*}
\mathbb P(M_n \leq n(1-\delta), B>n) \approx e^{-n(D(1/2 || p)}.
\end{equation*}
\end{lemma*}

We now turn to the initial $\frac{n}{2}$ terms in equation~\eqref{EqnMn}. We rewrite the probability as follows:
\begin{multline*}
\sum_{g = 0}^{n/2} \mathbb P(M_n \leq n(1-\delta), B \leq n, G = g) \\
= \sum_{g = 0}^{n/2} \P(G=g)\times
\left\{\sum_{b=2g}^n \sum_{a= -b}^b \P \left( \sum_{j=1}^{g} \tilde T_j = -a, \sum_{j=1}^{g} |\tilde T_j| = b  \right) \mathbbm{1}\left( \frac{a+b}{2} \ge  n\delta \right) \right\}.
\end{multline*}
This is because conditioned on $G = g$, the behavior of the random walk can be broken into segments in between return times to 0, corresponding to the sojourn times $\{\tilde T_1, \dots, \tilde T_g\}$. Furthermore, the time $B = b$ of the final sojourn time must be at least $2g$ and at most $n$, which is the time of the last return to 0; and the sum $-a$ of the signed sojourn times is then clearly in $[-b,b]$. Finally, the quantity $M_n$ is a sum of $(n-b)$ (corresponding to the final $n-b$ time steps) and the sum of lengths of positive sojourn segments, which is $\frac{b-a}{2}$. Then the inequality $M_n \le n(1-\delta)$ is equivalent to $\frac{a+b}{2} \ge n\delta$.

%This is because the total number of $+1$'s received equals $(n-b)$ (the $+1$'s received after time $b$) plus the total number of $+1$'s received up to time $b$, which equals $(b-a)/2$. This equals $n- (a+b)/2$, and since we are interested in the event that this quantity is at most $n(1-\delta)$, we introduce the indicator $\mathbbm{1}((a+b)/2 > n\delta)$.

We now substitute $\alpha \defn \frac{a}{n}, \beta \defn \frac{b}{n}$, and $\gamma \defn \frac{g}{n}$. We may rewrite the above expression as a summation over $\alpha, \beta$, and $\gamma$, where we implicitly assume that these variables take values of the form $\frac{i}{n}$, for some integer $i$:
\begin{equation*}
\sum_{\gamma = 0}^{1/2} \P\left(\frac{G}{n} = \gamma \right) \times \left\{\sum_{\beta=\max(\delta, 2\gamma)}^{1} \sum_{\alpha= 2\delta-\beta}^{\beta} \P \left( \frac{\sum_{j=1}^{n\gamma} \tilde T_j}{n} = -\alpha, \frac{\sum_{j=1}^{n\gamma} |\tilde T_j| }{n}= \beta  \right)\right\}.
\end{equation*}

Our next theorem is a core technical result of this paper:
\begin{theorem}
\label{thm: magic}
Let $\delta \in [0,1]$. The following equality holds:
\begin{align*}
\lim_{n \to \infty} \frac{1}{n} \log \P \left(\frac{M_n}{n} \leq (1-\delta)\right) = -\delta D(1/2|| p).
\end{align*}
\end{theorem}
\begin{proof}
Define the random vector $Z_n \defn \left(\frac{G}{n},  \frac{\sum_{j=1}^{G} \tilde T_j}{n},  \frac{\sum_{j=1}^{G} |\tilde T_j|}{n}\right)$. Define the set
\begin{equation*}
\cQ := \Big\{\alpha, \beta, \gamma \mid \gamma \in [0,1/2], \beta \in [\max(\delta, 2\gamma), 1], \alpha \in [2\delta-\beta, \beta]\Big\} \subseteq \real^3.
\end{equation*}
Note that we are interested in the quantity $\P(Z_n \in \cQ)$. We will evaluate this probability for large $n$ using the G\"{a}rtner-Ellis theorem from large deviation theory \cite{DupEll11}. The first step is to show that the following limit exists for every $\lambda \defn (\lambda_1, \lambda_2, \lambda_3) \in \mathbb R^3$:
\begin{align*}
&\Lambda(\lambda) \defn \lim_{n\to \infty} \frac{1}{n} \log \E e^{n\lambda \cdot Z_n}\\
 &= \lim_{n \to \infty} \frac{1}{n} \log \E \exp \left(\lambda_1G + \lambda_2 \left( \sum_{j=1}^{G} \tilde T_j\right) + \lambda_3 \left(\sum_{j=1}^{G} |\tilde T_j|\right) \right)\\
&= \lim_{n \to \infty} \frac{1}{n} \log \E \left[ \E \left[ \exp \left(\lambda_1G + \lambda_2 \left( \sum_{j=1}^{G} \tilde T_j\right) + \lambda_3 \left(\sum_{j=1}^{G} |\tilde T_j|\right) \right) \Big| G \right] \right]\\
&\stackrel{(a)}= \lim_{n \to \infty} \frac{1}{n} \log \E \left[ e^{\lambda_1G + L(\lambda_2, \lambda_3)G}\right]\\
&\stackrel{(b)}{=} 
\begin{cases}
0 &\text{ if } \lambda_1 + L(\lambda_2, \lambda_3) < \log\left(\frac{1}{2p}\right),\\
+\infty &\text{ otherwise},
\end{cases}
\end{align*}
where equality $(a)$ holds for $(\lambda_2, \lambda_3) \in \cD$, using the fact that the $\tilde T_j$'s are conditionally i.i.d.\ (and the limit is $+\infty$ if $(\lambda_2, \lambda_3) \notin \cD$); and equality $(b)$ holds by evaluating the moment generating function of a geometric random variable. To summarize, we have
\begin{align*}
\Lambda(\lambda) = 
\begin{cases}
+\infty &\text{ if } (\lambda_2, \lambda_3) \in \cD^c,\\
+\infty &\text{ if } (\lambda_2, \lambda_3) \in \cD \text{ and } \lambda_1 \geq \log\left(\frac{1}{2p}\right) - L(\lambda_2, \lambda_3),\\
0 &\text{ otherwise.} 
\end{cases}
\end{align*}
Let the domain of $\Lambda$ be $\cD_\Lambda$, and let $\Lambda^*$ denote the convex conjugate of $\Lambda$. A direct application of the G\"{a}rtner-Ellis theorem then gives the following upper bound:
\begin{align}
\label{eq: upper}
\limsup_{n\to \infty} \frac{1}{n} \log \P(Z_n \in \cQ) \leq - \inf_{z \in \cQ} \Lambda^*(z).
%&\liminf_{n \to \infty} \frac{1}{n} \log \P_n(\cR^{\mathrm{o}}) \geq -\inf_{\cR^{\mathrm{o}} \cap \cE} \Lambda^*(z)
\end{align}
%where $\P_n$ is the distribution of $Z_n$.
We now evaluate the convex conjugate of $\Lambda$ at the point $(\gamma, \alpha, \beta) \in \real_+^3$:
\begin{align*}
\Lambda^*(\gamma, \alpha, \beta) & = \sup_{\lambda \in \cD_\Lambda} \left\{ \lambda_1 \gamma + \lambda_2 \alpha + \lambda_3 \beta\right\} \\
&\stackrel{(a)}= \sup_{(\lambda_2, \lambda_3) \in \cD} \left\{\left( \log\left(\frac{1}{2p}\right) - L(\lambda_2, \lambda_3)\right)\gamma + \lambda_2 \alpha + \lambda_3 \beta\right\},
\end{align*}
where in equality $(a)$, we have used the fact that $\gamma \geq 0$ in $\cQ$.  We now make the following crucial observations:  First, Lemma~\ref{lemma: gamma} implies that the coefficient of $\gamma$ in the above expression is positive, so $\Lambda^*(\gamma, \alpha, \beta)$ is a monotonically increasing function of $\gamma$ for fixed $\alpha$ and $\beta$. Second, the set $\cQ$ is such that the possible values of the pair $(\alpha, \beta)$ can only increase as $\gamma$ becomes smaller. This implies that if $\gamma_1 < \gamma_2$, then 

\begin{align*}
\inf_{(\alpha, \beta) : (\gamma_1, \alpha, \beta) \in \cQ} \Lambda^*(\gamma_1, \alpha, \beta)  < \inf_{(\alpha, \beta) : (\gamma_2, \alpha, \beta) \in \cQ} \Lambda^*(\gamma_2, \alpha, \beta) ,
\end{align*}
since not only is the left-hand objective smaller than the right-hand objective for every fixed $(\alpha,\beta)$, but also the range of possible values of $(\alpha, \beta)$ on the left-hand side contains the range of values on the right-hand side.
Thus, the infimum of $\Lambda^*$ over $\cQ$ must occur when $\gamma = 0$; i.e.,
\begin{align*}
\inf_{(\gamma, \alpha, \beta) \in \cQ} \Lambda^*(\gamma, \alpha, \beta) &= \inf_{\beta \in [\delta, 1], \alpha \in [2\delta-\beta, \beta]} \Lambda^*(0, \alpha, \beta).
\end{align*}
When $\beta \in [\delta, 1]$ and $\alpha \in [\beta-2\delta, \beta]$, we have
\begin{align*}
\Lambda^*(0, \alpha, \beta) &= \sup_{(\lambda_2, \lambda_3) \in \cD} \left\{\lambda_2 \alpha+ \lambda_3 \beta\right\}
= \beta D(1/2 || p).
\end{align*}
Hence, we conclude that
\begin{align*}
\inf_{(\gamma, \alpha, \beta) \in \cQ} \Lambda^*(\gamma, \alpha, \beta) &= \inf_{\beta \in [\delta, 1], \alpha \in [2\delta-\beta, \beta]} \beta D(1/2 || p)\\
&= \delta D(1/2 || p).
\end{align*}

To complete the proof, we need to establish a lower bound counterpart to inequality \eqref{eq: upper}. This is established via the following lemma, proved in Appendix \ref{proof: lemma: lobo}:
\begin{lemma*}
\label{lemma: lobo}
The following inequality holds:
\begin{align*}
\liminf_{n \to \infty} \frac{1}{n} \log \P(Z_n \in \cQ) \geq -\delta D(1/2 || p).
\end{align*}
\end{lemma*}
The proof follows by constructing a set of paths that satisfy $M_n < (1-\delta)n$ and explicitly computing the combined probability of these paths.
This completes the proof of Theorem \ref{thm: magic}.
%and we conclude $
%\P \left(\frac{M_n}{n} \leq (1-\delta)\right) \approx e^{-n\delta D(1/2 || p)}.$
\end{proof}

%We find it remarkable that the calculations above can be performed to completion, and that the final expression has such a simple form.

%%%%%

\subsection{Large deviations for a Bernoulli mixture}

Finally, we first prove the following large deviations result for a mixture of Bernoulli random variables. Recall~\cite{DupEll11} that the rate function $I$ of a sequence of random variables $\{W_n\}$ has the property that for any closed subset $F \subseteq \real$ and any open subset $G \subseteq \real$, we have
\begin{align*}
\lim\sup_{n \rightarrow \infty} \frac{1}{n} \log \mprob(W_n \in F) & \le - \inf_{w \in F} I(w), \\
\lim\inf_{n \rightarrow \infty} \frac{1}{n} \log \mprob(W_n \in G) & \ge - \inf_{w \in G} I(w).
\end{align*}

\begin{lemma*}
\label{lemma: cricket}
Let $\theta \in [0,1]$ and $q \in [0,1/2]$. Consider a sequence of i.i.d.\ Bernoulli($1-q$) random variables $\{U_i\}_{1 \leq i \leq n- \lfloor n\theta \rfloor}$ and a sequence of i.i.d.\ Bernoulli($q$) random variables $\{V_j\}_{1 \leq j \leq \lfloor n\theta \rfloor}$, such that the $U_i$'s are independent of the $V_j$'s. The random variable
\begin{align*}
W_n \defn \frac{\sum_{i=1}^{n - \lfloor n\theta \rfloor} U_i + \sum_{j=1}^{\lfloor n\theta \rfloor} V_j}{n}
\end{align*}
satisfies the large deviation principle with rate function 
\begin{align*}
&I_\theta(w) = w\log \left(  \eta \right)
- \bar \theta \log \left( \bar q \eta + q\right)
- \theta \log \left(q\eta + \bar q\right),
\end{align*}
where
\begin{align*}
\eta := \frac{-\tau + \sqrt{\tau^2+4w\bar w}}{2\bar w}, \quad \text{and} \quad \tau \defn  \frac{\bar q}{q}(\bar \theta - w) + \frac{q}{\bar q}(\theta - w).
\end{align*}
\end{lemma*}
The proof is a direct application of the G\"{a}rtner-Ellis theorem and is detailed in Appendix \ref{proof: lemma: cricket}. Lemma~\ref{lemma: cricket} will be useful for evaluating the probability that the student makes an error when we condition on various aspects of the random walk $\{X_n\}$, such as the number $M_n$ of +1 transmissions by the teacher, or the last return time of the random walk to 0.

%%%%%

\section{Majority learning}
\label{SecMajority}

We now present our main results in the case where the student employs the simplest majority rule strategy.
%Assume, without loss of generality, that the state of the world is $\Theta = +1$.

\subsection{Simple forwarding}

If the teacher employs the simple forwarding strategy, the student's observations $\{Z_n\}$ may be viewed as noisy observations of $\Theta$ through a binary symmetric channel with error parameter
\begin{equation*}
p\star q := p(1-q) + q(1-p) = p\bar q + q \bar p.
\end{equation*}
In this case, the majority learning strategy for the student exactly corresponds to the Bayesian strategy, producing an optimal learning rate of $D(1/2 || p\star q)$~\cite{DupEll11}. We will use this simple expression as a benchmark for comparing the learning rate of more complicated combinations of teacher/student strategies.

%%%%%

\subsection{Cumulative teaching}
\label{section: analysis}

Note that Theorem \ref{thm: magic} already provides the exact learning rate if $q=0$, since the student makes an error exactly when $M_n < \frac{n}{2}$. Substituting $\delta = 1/2$, we see that the student will learn at a rate of $\frac{1}{2} D(1/2 || p)$ via the cumulative teaching strategy. In contrast, the learning rate for simple forwarding is $D(1/2 || p)$, since $p \star q = p$---and this rate is \emph{higher} than that of cumulative teaching! It is natural to wonder what values of the channel parameters $(p,q)$ make simple forwarding a better strategy than cumulative teaching, and vice versa.

%Indeed, conditioned on the fact that the last return time is $B = \lfloor n \theta\rfloor$, we know that the teacher must transmit only +1's for the remaining $n-\lfloor n \theta \rfloor$ time instances.

To analyze the general case $q > 0$, we can use Lemma~\ref{lemma: cricket} to evaluate the probability that at most $n(1-\delta)$ instances of the student's received sequence $\{Z_n\}$ are equal to $+1$. The learning rate of a student using a majority learning rule can then be obtained by plugging in $\delta = \frac{1}{2}$. 

\begin{theorem}
\label{ThmCricket}
Let $\delta \in [q, 1-q]$. Suppose we say that the student commits an error if the fraction of received $+1$'s is at most $(1-\delta)$. Then the rate of learning is given by 
\begin{align*}
\cR = \inf_{\theta \in [0,1]} \left\{ \theta D(1/2 || p) + \hat I(\theta)\right\},
\end{align*}
where
\begin{align*}
\hat I(\theta) = \inf_{w \in [0, 1-\delta]} I_\theta(w),
\end{align*}
and $I_\theta(\cdot)$ is the rate function appearing in Lemma~\ref{lemma: cricket}.
\end{theorem}
The learning rate provided in Thereom~\ref{ThmCricket} can be approximately calculated using computing software; see Figure~\ref{fig: corona}. 
\begin{proof}
Let $\cE_n$ be the error event that the number of $+1$'s received by the student is at most $n(1-\delta)$. Consider an integer $N > 0$, whose value will be specified later. We divide the interval $[0,1)$ into intervals $\{L_i\}_{i=0}^{N-1}$, where $L_i \defn \left[\frac{i}{N}, \frac{i+1}{N}\right)$. The probability of error can then be written as
\begin{align}
\label{EqnProbErr}
\P(\cE_n) &= \sum_{i=0}^{N-1} \P\left(\frac{M_n}{n} \in L_i\right) \P\left(\cE_n \Big | \frac{M_n}{n} \in L_i\right).
\end{align}

Note that
\begin{align}
\label{EqnMnExpand}
\P\left(\frac{M_n}{n} \in L_i\right) = \P\left(\frac{M_n}{n} < \frac{i+1}{N}\right) - \P\left(\frac{M_n}{n} < \frac{i}{N}\right).
\end{align}
Let $\epsilon_1 > 0$ be an arbitrarily small constant. By Theorem~\ref{thm: magic}, we know that for all $0 \leq i \leq N-1$ and for all large enough $n$, the following inequality holds:
\begin{align}\label{eq: ML}
\Big|- \frac{1}{n} \log \P\left(\frac{M_n}{n} \in L_i\right) - \frac{N-i-1}{N} D\left(1/2 || p \right)\Big| < \epsilon_1,
\end{align}
since the first term in equation~\eqref{EqnMnExpand} dominates.

%The probability of error may be bounded from both sides via the relation
%\begin{align*}
%\sum_{i=0}^{N-1} e^{-n\left( \frac{N-i-1}{N} D(1/2 || p) + \epsilon_1 \right)} \P\left(\cE_n \Big | \frac{M_n}{n} \in L_i\right) \leq \P(\cE_n) &\leq \sum_{i=0}^{N-1} e^{-n\left( \frac{N-i-1}{N} D(1/2 || p) - \epsilon_1 \right)} \P\left(\cE_n \Big | \frac{M_n}{n} \in L_i\right).
%\end{align*}
Turning to the second term in the summand of equation~\eqref{EqnProbErr}, we note that the probability of error is a monotonically decreasing function of $\frac{M_n}{n}$, so
\begin{align}
\label{EqnEn}
\P\left(\cE_n \Big | M_n = \left\lceil n \cdot \frac{i+1}{N} \right\rceil \right) \leq \P\left(\cE_n \Big | \frac{M_n}{n} \in L_i\right) \leq \P\left(\cE_n \Big | M_n = \left\lfloor n \cdot \frac{i}{N}\right\rfloor\right).
\end{align}
Let $\epsilon_2>0$ be an arbitrarily small constant. Using the large deviations principle from Lemma~\ref{lemma: cricket}, we know that for all large enough $n$ and for all $0 \leq i \leq N-1$, we have the bound
\begin{align}\label{eq: M}
\Big| \frac{1}{n} \log \P\left(\cE_n \Big | M_n= \left\lfloor n \cdot \frac{i}{N}\right\rfloor \right) - \hat I\left(\frac{N-i}{N}\right) \Big| < \epsilon_2,
\end{align}
since conditioned on the fraction $\frac{M_n}{n} \approx \frac{i}{N}$ of +1's that are transmitted by the teacher, the distribution of +1's received by the student behaves exactly as the mixture of Bernoulli distributions in Lemma~\ref{lemma: cricket} with $\theta = 1 - \frac{i}{N}$. Furthermore, the conditional probability of the error event exactly corresponds to the infimum of the rate function over the interval $[0,1-\delta]$. We have a similar bound for the left-hand expression in inequality~\eqref{EqnEn}.

Combining the bounds from inequalities~\eqref{eq: ML} and \eqref{eq: M}, we then obtain
\begin{align*}
\P(\cE_n) &\leq \sum_{i=0}^{N-1} e^{-n\left( \frac{N-i-1}{N} D(1/2 || p) + \hat I\left(\frac{N-i}{N}\right)- \epsilon_1 -\epsilon_2\right)}, \\
\P(\cE_n) &\geq \sum_{i=0}^{N-1} e^{-n\left( \frac{N-i-1}{N} D(1/2 || p) + \hat I\left(\frac{N-i-1}{N}\right) +\epsilon_1 +\epsilon_2\right)}.
\end{align*}
Let $\epsilon_3 > 0$ be an arbitrarily small constant. Define the three quantities
\begin{align*}
u &\defn \inf_{x \in [0,1]} \left\{xD(1/2 || p) + \hat I(x)\right\}, \\
\bar u &\defn \inf_{0 \leq i \leq N-1} \left\{ \frac{N-i-1}{N} D(1/2 || p) + \hat I\left(\frac{N-i}{N}\right)\right\}, \\
\underbar{$u$} &\defn \inf_{1 \leq i \leq N-1} \left\{\frac{N-i-1}{N} D(1/2 || p) + \hat I\left(\frac{N-i-1}{N}\right)\right\}.
\end{align*}
Using the continuity of $\hat I$, we can now pick $N$ (depending only on $\epsilon_3$ and $p$) such that 
\begin{align*}
\max(|u - \bar u|, |u - \underbar{$u$}|) < \epsilon_3.
\end{align*}
Then for all large enough $n$, we have the bounds
\begin{align*}
\P(\cE_n) &\leq N e^{-n(u-\epsilon_1-\epsilon_2-\epsilon_3)},\\
\P(\cE_n) &\geq e^{-n(u+\epsilon_1+\epsilon_2+\epsilon_3)}.
\end{align*}
Taking logarithms, dividing by $n$, and taking the limit, we see that
\begin{align*}
\Big| \lim_{n \to \infty} \frac{1}{n} \log \P(\cE_n) - u \Big| < \epsilon_1+\epsilon_2+\epsilon_3.
\end{align*}
Since $\epsilon_1, \epsilon_2$, and $\epsilon_3$ are arbitrary constants, we conclude that 
\begin{align*}
\lim_{n \to \infty} \frac{1}{n} \log \P(\cE_n) = u,
\end{align*}
completing the proof.
\end{proof}

\begin{remark*}
The learning rate in Theorem~\ref{ThmCricket} may be simplified further: Note that when $\bar \theta \bar q + \theta q \in [0, 1-\delta]$, the value of $\hat I(\theta)$ is 0, since $I_\theta(\bar \theta \bar q + \theta q) = 0$. The constraint $\bar \theta \bar q + \theta q \in [0, 1-\delta]$ is equivalent to $\theta \in \left[\frac{\bar q}{\bar q - q}, \frac{\delta - q}{\bar q - q}\right]$. Hence,
the minimum of $\theta D(1/2 || p) + \hat I(\theta)$ is clearly achieved over this range for $\theta = \frac{\delta - q}{\bar q - q}$. Further note that for $\theta \in \left[0, \frac{\delta-q}{\bar q -q}\right]$, we have the equality 
$\hat I(\theta) = I_\theta(1-\delta)$, since $I_\theta$ is convex and minimized at $w = \bar \theta \bar q + \theta q > 1-\delta$. Thus, we can rewrite the learning rate as 
\begin{align*}
\cR = \inf_{\theta \in \left[0, \frac{\delta-q}{\bar q -q}\right]} \Big\{\theta D(1/2 || p) + I_\theta(1-\delta)\Big\}.
\end{align*}
This expression does not simplify further; but since the function being minimized is known in closed form, the value of $\cR$ is easy to simulate using Matlab or similar software.
\end{remark*}

%%%%%

\section{$\epsilon$-majority learning}
\label{SecEpsMajority}

In the previous section, we assumed that the student is simply a majority learner; i.e., the student takes the majority of her observations to ultimately decide the value of $\hat \Theta$. We now explore the alternative strategy where the student's final estimate is computed to be the majority over the final $\epsilon n$ observations, corresponding to a more lazy or more skeptical student. We assume that the student has access to the problem parameters $p$ and $q$, and utilizes the value of $\epsilon$ that maximizes her learning rate.

%%%%%

\subsection{Simple forwarding}

Suppose that the teacher is simply forwarding his observations to the student and the student is taking an $\epsilon$-majority based on her observations. It is easy to see that the learning rate for the student in this case is simply $\epsilon D(1/2 || p \star q)$, since her estimate is always a majority over some number of i.i.d.\ observations. Thus, the optimal choice of $\epsilon$ is 1, and the best learning rate for $\epsilon$-majority learning simply coincides with majority learning.

%%%%%

\subsection{$\epsilon$-teaching}

We first analyze the $\epsilon$-teaching strategy, where the teacher transmits noise in the first $(1-\epsilon)n$ time steps, and then repeatedly transmits his best guess at that point in the remaining $\epsilon n$ time steps.

\begin{theorem}
\label{thm: eps_teaching}
For the $\epsilon$-teaching strategy, the student's learning rate, optimized over all values of $\epsilon \in [0,1]$, is given by
\begin{equation*}
\cR = \frac{D(1/2||p) D(1/2 ||q)}{D(1/2 || p) + D(1/2 || q)}.
\end{equation*}
\end{theorem}

\begin{proof}
Let the teacher's estimate at time $\bar \epsilon n$, obtained by taking a majority over his observations until that time, be $\tilde \Theta$. The teacher then transmits $\tilde \Theta$ repeatedly for the final $\epsilon n$ time slots, and the student takes a majority over her $\epsilon n$ observations to determine $\hat \Theta$. The probability of error is thus calculated to be
\begin{align}
\label{EqnErrTheta}
\P(\hat \Theta \neq \Theta) &= \P(\hat \Theta \neq \Theta, \tilde \Theta \neq \Theta) + \P(\hat \Theta \neq \Theta, \tilde \Theta = \Theta) \notag \\
&= \P(\hat \Theta \neq \Theta | \tilde \Theta \neq \Theta)\P(\tilde \Theta \neq \Theta) + \P(\hat \Theta \neq  \Theta | \tilde \Theta = \Theta)\P(\tilde \Theta = \Theta) \notag \\
&= \P(\hat \Theta = \tilde \Theta)\P(\tilde \Theta \neq \Theta) + \P(\hat \Theta \neq  \tilde \Theta)\P(\tilde \Theta = \Theta).
\end{align}
Notice that
\begin{align*}
\P(\tilde \Theta \neq \Theta) \approx \exp(-n\bar\epsilon D(1/2 || p)),
\end{align*}
and 
\begin{align*}
\P(\hat \Theta \neq \tilde \Theta) \approx \exp(-n\epsilon D(1/2 || q)),
\end{align*}
since the teacher obtains $\tilde \Theta$ from $n\bar \epsilon$ observations of $\Theta$ through a binary symmetric channel with error probability $p$, and the student obtains $\hat \Theta$ from $n\epsilon$ observations of $\tilde \Theta$ through a binary symmetric channel with error probability $q$.

Substituting back into equation~\eqref{EqnErrTheta}, we see that the learning rate is determined by
\begin{align*}
\min(\bar \epsilon D(1/2 || p), \epsilon D(1/2 || q)).
\end{align*}
It is not hard to see that this expression is maximized when 
\begin{align*}
\epsilon = \frac{D(1/2 || p)}{D(1/2 || p)+D(1/2 || q)},
\end{align*}
yielding the specified learning rate.
%and the resulting learning rate is 
%\begin{align*}
%\frac{D(1/2 || p)D(1/2 || q)}{D(1/2 || p)+D(1/2 || q)}.
%\end{align*}
\end{proof}

As we will see, the rate obtained in Theorem~\ref{thm: eps_teaching} is not optimal from the point of view of the teacher; however, the $\epsilon$-teaching strategy nonetheless provides a nice closed-form expression for the learning rate, which is convenient for comparison.

\begin{figure}
\begin{center}
\begin{tabular}{cc}
\includegraphics[width=3in]{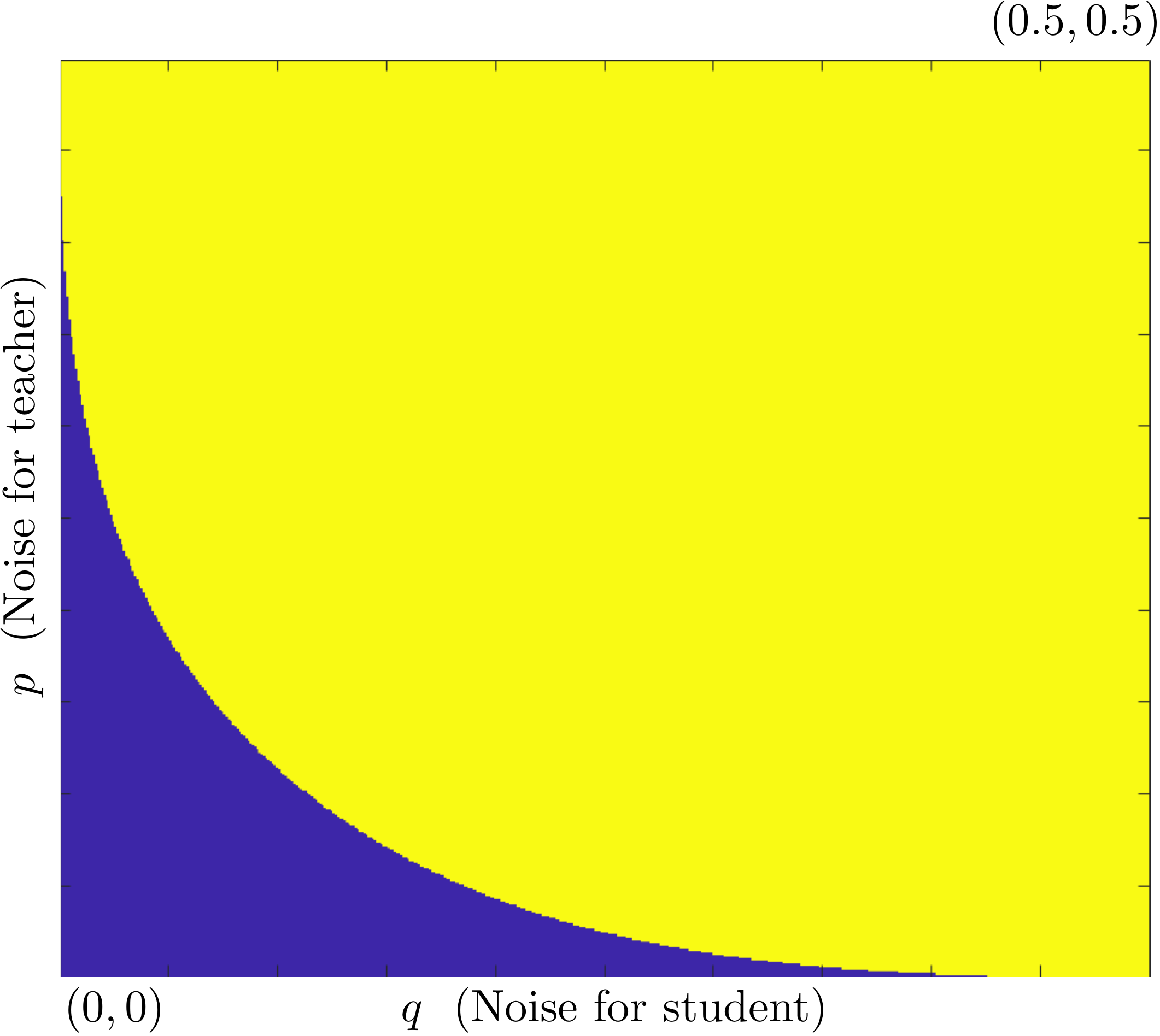} & \includegraphics[width=3in]{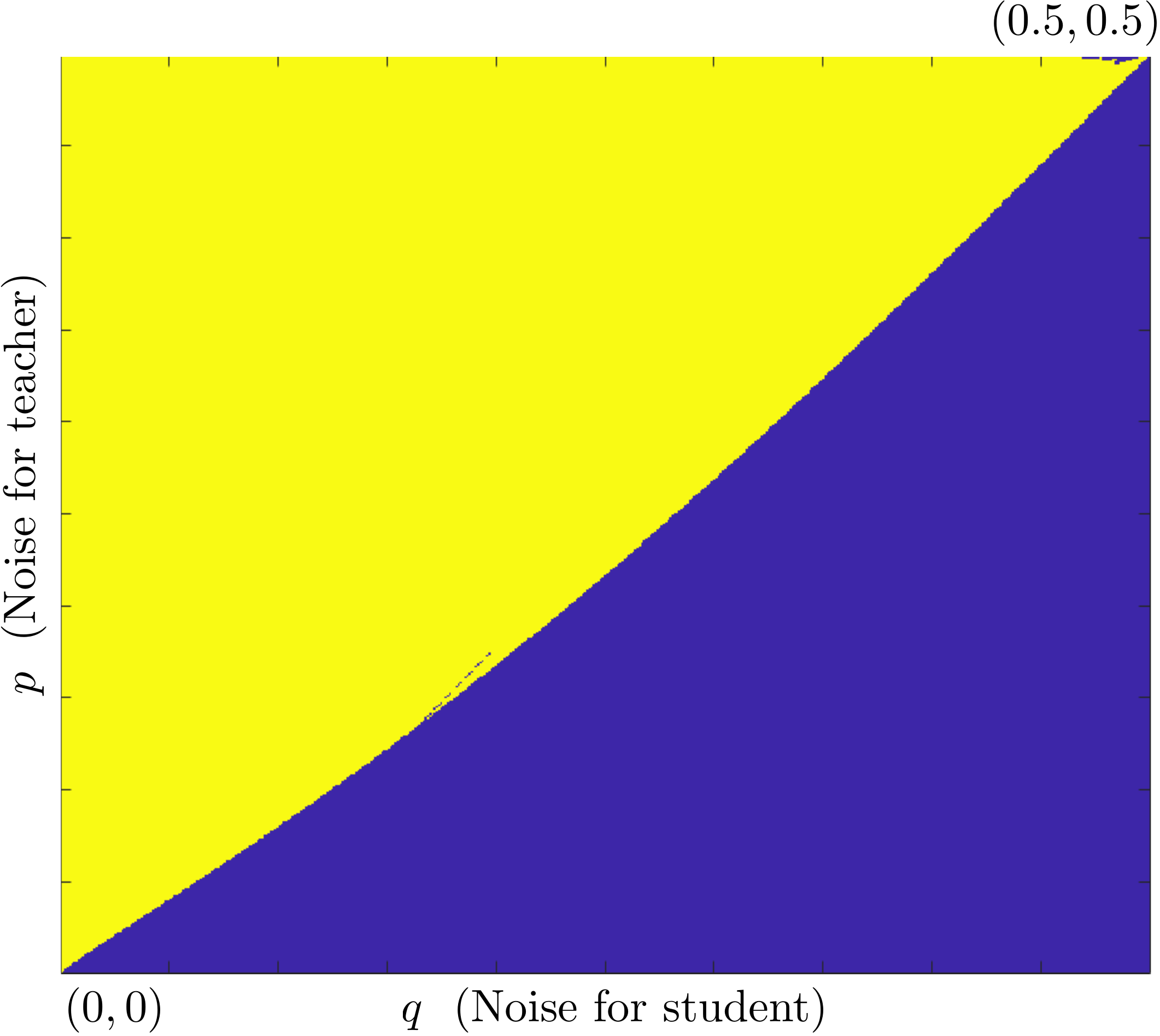} \\
(a) $\epsilon$-teaching + $\epsilon$-learning (yellow) vs.\ & (b) $\epsilon$-teaching + $\epsilon$-learning (yellow) vs.\ \\
simple forwarding + majority learning (blue) & cumulative teaching + majority learning (blue)
\end{tabular}
\caption{Plots showing the benefits of $\epsilon$-teaching + $\epsilon$-learning, analyzed in Theorem~\ref{thm: eps_teaching}. The figure on the right is simulated using the formula in Theorem \ref{ThmCricket}. In both figures, the yellow region indicates the parameter values where $\epsilon$-teaching + $\epsilon$-learning dominates the competing strategy based on majority learning. In (a), note that when the cumulative noise in the communication system is small, it is better to simply forward information. In (b), note that when the teacher's noise is large compared to the student's noise, the $\epsilon$-teaching strategy dominates the cumulative teaching strategy, since the teacher has less time to confuse himself.} \label{fig: corona}
\end{center}
\end{figure}

%%%%%

\subsection{Cumulative teaching}

In the above strategy, the teacher essentially stops learning after time $\bar \epsilon n$. Intuitively, the learning rate should be higher if the teacher continues to transmit better and better estimates of the state of the world in the latter $\epsilon n$ time steps, based on his own continual learning. Consequently, we now analyze the cumulative teaching strategy of the teacher, where the student continues to be an $\epsilon$-majority learner. Note that when $\epsilon = 1$, this strategy is identical to the cumulative teaching + majority learning combination studied before. Since the best choice of $\epsilon$ will outperform $\epsilon = 1$, the learning rate for this combination will be at least as high.

We claim that when the student is an $\epsilon$-learner, cumulative teaching is indeed at least as good as $\epsilon$-teaching. This is stated and proved as a warm-up in Proposition~\ref{PropCumEps}; in Theorem~\ref{ThmOptCum} to follow, we provide an expression for the precise learning rate.

\begin{proposition}	
\label{PropCumEps}
The learning rate of cumulative teaching and $\epsilon$-majority learning (with the optimum choice of $\epsilon$) is at least as large as the learning rate of the $\epsilon$-teaching strategy from Theorem \ref{thm: eps_teaching}.
\end{proposition}

\begin{proof}
Fix $\epsilon > 0$, and let $B$ be the final time that the random walk $\{X_n\}$ visits 0. Let $\hat \Theta$ denote the final estimate of the student. We write
\begin{equation*}
\mprob(\hat \Theta \neq \Theta) = \mprob(\hat \Theta \neq \Theta \mid B \le \bar \epsilon n) \mprob(B \le \bar \epsilon n) + \mprob(\hat \Theta \neq \Theta \mid B > \bar \epsilon n) \mprob(B > \bar \epsilon n),
\end{equation*}
and compare the respective quantities in the cases of cumulative teaching vs.\ $\epsilon$-teaching. The probabilities $\mprob(B \le \bar \epsilon n)$ and $\mprob(B > \bar \epsilon n)$ have no dependence on the strategy employed by the teacher.

Note that conditioned on the event $\{B \leq \bar \epsilon n\}$, the $\epsilon$-teaching and cumulative teaching strategies are identical as far as the final $\epsilon n$ bits are concerned, so the values of $\mprob(\hat \Theta \neq \Theta \mid B \le \bar \epsilon n)$ are identical. Conditioned on the event $\{B > \bar \epsilon n\}$, the random walk $\{X_n\}$ is equally likely to be positive or negative at time $\bar \epsilon n$, so in the case of $\epsilon$-teaching, we have $\mprob(\hat \Theta \neq \Theta \mid B > \bar \epsilon n) = \frac{1}{2}$. In the case of cumulative teaching, we can lower-bound the learning rate by upper-bounding the error as $\mprob(\hat \Theta \neq \Theta \mid B > \bar \epsilon n) \le  1$.

It is easy to see that increasing the value of the error probability $\mprob(\hat \Theta \neq \Theta \mid B > \bar \epsilon n)$ by a factor of 2 cannot make the learning rate of cumulative teaching worse than the learning rate of $\epsilon$-teaching. Thus, we conclude that for every choice of $\epsilon > 0$, the learning rate of cumulative teaching is at least as high as the learning rate of $\epsilon$-teaching. Optimizing the student's choice of $\epsilon$ when the teacher employs cumulative teaching can only increase the learning rate further, completing the proof.
\end{proof}

Using a more sophisticated argument, we can obtain an expression for the learning rate of cumulative teaching when the student is an $\epsilon$-majority learner:

\begin{theorem}
\label{ThmOptCum}
The optimal learning rate for the cumulative teaching and $\epsilon$-learning strategy is
\begin{align*}
\cR = \sup_{\epsilon \in [0,1]} \left[ \min\left( \epsilon D(1/2 || q),\inf_{\alpha \in [0,1/2]} \left\{(\bar \epsilon +  \epsilon \alpha)D(1/2 || p)+ \epsilon I_\alpha(1/2)\right\}\right)\right],
\end{align*}
where $I_\alpha(\cdot)$ is the rate function defined in Lemma~\ref{lemma: cricket}.
\end{theorem}

\begin{proof}
We will show that the learning rate for a fixed value of $\epsilon$ is given by the expression inside the objective function. Let $B$ denote the last time at which the random walk $\{X_n\}$ corresponding to the teacher's transmissions visits 0. We write the error probability of the student as
\begin{align}
\label{EqnErrProb}
\P(\hat \Theta \neq \Theta) = \P(\hat \Theta \neq \Theta, B\leq (1-\epsilon)n) + \P(\hat \Theta \neq \Theta, (1-\epsilon)n < B \leq n) + \P(\hat \Theta \neq \Theta, B>n). 
\end{align}
Naturally, the largest of the three error probabilities determines the overall learning rate. For the third term in equation~\eqref{EqnErrProb}, we observe that
\begin{align*}
\P(\hat \Theta \neq \Theta, B>n) &= \P(B > n) P(\hat \Theta \neq \Theta | B>n)\\
&= \frac{1}{2} \P(B > n) \\
&\approx \exp(-n D(1/2 || p)),
\end{align*}
where the second equality holds by the symmetry of the random walk conditioned on $\{B > n\}$, and the final approximation follows from the proof of Lemma~\ref{LemMnSecond} from Appendix~\ref{AppLemMnSecond}. For the first term in equation~\eqref{EqnErrProb}, notice that if $B \le n(1-\epsilon)$, the teacher will transmit $+1$'s for the entire duration over which the student is taking a majority. The student's error probability is the probability of receiving more than $\frac{n\epsilon}{2}$ observations of $-1$'s in the final $n\epsilon$ observations, so we have
\begin{align*}
\P(\hat \Theta \neq \Theta, B\leq (1-\epsilon)n) &= \P(\hat \Theta \neq \Theta | B\leq (1-\epsilon)n) \cdot \P( B\leq (1-\epsilon)n )\\
&\approx \exp(-n\epsilon D(1/2 || q)),
\end{align*}
using the fact that $\mprob(B > (1-\epsilon)n) \approx \exp(-(1-\epsilon)n D(1/2 \| p))$, so $\mprob(B \le (1-\epsilon)n ) = 1 - o(1)$.

Turning to the middle term in equation~\eqref{EqnErrProb}, suppose we fix $\alpha \in [0,1]$, and suppose $B = \lfloor (1-\epsilon)n + \alpha \epsilon n \rfloor$. As we will see, the error probability is dominated by the error probability of a random walk that is negative at all time instances up to $B$. Thus, we define the indicator variable $I = 1$ if $\{X_n\}$ is negative for all time instances up to $B$, and 0 otherwise. We then write
\begin{align*}
\P(B = \lfloor n\bar \epsilon + n\epsilon \alpha \rfloor, \hat \Theta \neq \Theta) & = \mprob(\hat \Theta \neq \Theta \mid B = n_\alpha, I = 0) \mprob(B = n_\alpha, I = 0) \\
& \qquad + \mprob(\hat \Theta \neq \Theta \mid B = n_\alpha, I = 1) \mprob(B = n_\alpha, I = 1),
\end{align*}
where we have used the shorthand $n_\alpha := \lfloor n \bar \epsilon + n \epsilon \alpha \rfloor$. Note that
\begin{align*}
\mprob(B = n_\alpha) & = (1-2p){n_\alpha \choose n_\alpha/2} p^{n_\alpha/2} \bar p^{n_\alpha/2}, \\
\mprob(B = n_\alpha, I = 1) & = (1-2p) \frac{1}{n_\alpha/2} {n_\alpha - 2 \choose n_\alpha/2 - 1} p^{n_\alpha/2} \bar p^{n_\alpha/2},
\end{align*}
where the second expression involves the corresponding Catalan number. Furthermore, by the same argument used in the proof of Lemma~\ref{lemma: lobo}, we know that
\begin{equation}
\label{EqnBna}
\mprob(B = n_\alpha, I = 1) \approx \exp(-n(\bar \epsilon + \epsilon \alpha) D(1/2 \| p)).
\end{equation}
Since $\mprob(B = n_\alpha)$ is at most a polynomial factor of $n$ larger, we conclude that
\begin{equation*}
\mprob(B = n_\alpha, I = 0) \approx \exp(-n(\bar \epsilon + \epsilon \alpha) D(1/2 \| p)),
\end{equation*}
as well. Turning to the conditional probability terms, note that
\begin{equation*}
 \mprob(\hat \Theta \neq \Theta \mid B = n_\alpha, I = 0) \le  \mprob(\hat \Theta \neq \Theta \mid B = n_\alpha, I = 1).
\end{equation*}
Hence, the error probability $\P(B = n_\alpha, \hat \Theta \neq \Theta)$ is within a poly$(n)$ factor of $\mprob(\hat \Theta \neq \Theta \mid B = n_\alpha, I = 1) \mprob(B = n_\alpha, I = 1)$, implying that
\begin{equation*}
\lim_{n \rightarrow \infty} \frac{1}{n} \log \mprob(B = n_\alpha, \hat \Theta \neq \Theta) = \lim_{n \rightarrow \infty} \frac{1}{n} \log \left(\mprob(\hat \Theta \neq \Theta \mid B = n_\alpha, I = 1) \mprob(B = n_\alpha, I = 1)\right).
\end{equation*}

Now suppose $\alpha \in [0,1/2]$. We have
\begin{equation*}
\mprob(\hat \Theta \neq \Theta \mid B = n_\alpha, I = 1) \approx \exp(-n\epsilon I_\alpha(1/2)),
\end{equation*}
since conditioned on the event $\{B = n_\alpha, I = 1\}$, we know that the teacher transmits $\epsilon \alpha n$ values that are -1 and $(\epsilon n - \epsilon \alpha n)$ values that are +1 to the student in the last $\epsilon n$ time steps, so the error probability is given by the rate function in Lemma~\ref{lemma: cricket}. Combining the equations, we have
\begin{align*}
\mprob(B = n_\alpha, \hat \Theta \neq \Theta) & \approx \exp(-n(\bar \epsilon +  \epsilon \alpha)D(1/2 || p)) \cdot \exp(-n \epsilon I_\alpha(1/2))\\
&= \exp( -n((\bar \epsilon +  \epsilon \alpha)D(1/2 || p)+ \epsilon I_\alpha(1/2) )).
\end{align*}
Finally, for $\alpha > \frac{1}{2}$, note that
\begin{equation*}
\mprob(B = n_\alpha, I = 1) \le \mprob(B = n_{1/2}, I = 1)
\end{equation*}
by equation~\eqref{EqnBna}, and
\begin{equation*}
\mprob(\hat \Theta \neq \Theta \mid B = n_\alpha, I = 1) \le 1 = 2 \mprob(\hat \Theta \neq \Theta \mid B = n_{1/2}, I = 1).
\end{equation*}
Thus, we conclude that the error rate is given by
\begin{equation*}
\inf_{\alpha \in [0,1]} \lim_{n \rightarrow \infty} \frac{-1}{n} \log \mprob(B = n_\alpha, \hat \Theta \neq \Theta) = \inf_{\alpha \in [0,1/2]} \left\{(\bar \epsilon +  \epsilon \alpha)D(1/2 || p)+ \epsilon I_\alpha(1/2)\right\}.
\end{equation*}

%Note that the worst-case error event is when the random walk $\{X_n\}$ is entirely negative for all time instants until $B$, for two reasons:
%\begin{itemize}
%\item[(i)] Conditioned on $B$, an entirely negative path happens with probability $\Theta(1/n)$, which is as good as probability 1 in the exponential scale.
%\item[(ii)] This path gives the maximum number of $-1$'s in the final $\epsilon n$ transmissions, which leads to the largest error probability for the student. The probability of such a path is approximately $\exp(-n(\bar \epsilon +  \epsilon \alpha)D(1/2 || p))$. Conditioned on this event, notice that the teacher's transmissions in the final $\epsilon n$ time slots consist of $n \epsilon \alpha$ transmissions of $-1$'s and $n\epsilon \bar \alpha$ transmissions of $+1$'s. Each of these transmitted bits gets flipped with probability $q$, and the student makes an error if she receives a majority of $-1$'s.
%\end{itemize}
%If $\alpha > 1/2$, then the probability that the student makes an error is at least $1/2$. Thus, we may assume $\alpha \in [0,1/2]$, since the error event $\alpha = 1/2$ dominates the error events corresponding to $\alpha > 1/2$. Assuming $\alpha \in [0, 1/2]$, the error probability can be calculated directly using Lemma \ref{lemma: cricket} as $\exp(-n\epsilon I_\alpha(1/2))$, if $\alpha \in [0,1/2]$. Thus,

Combining the bounds for the three terms in equation~\eqref{EqnErrProb}, the overall learning rate for a fixed choice of $\epsilon$ is therefore equal to
\begin{align*}
\min\left( D(1/2 || p), \epsilon D(1/2 || q),\inf_{\alpha \in [0,1/2]} \left\{(\bar \epsilon +  \epsilon \alpha)D(1/2 || p)+ \epsilon I_\alpha(1/2)\right\}\right). 
\end{align*}
Since the student is allowed to tune $\epsilon$, the optimal learning rate is therefore
\begin{align*}
\sup_{\epsilon \in [0,1]} \left[ \min\left( D(1/2 || p), \epsilon D(1/2 || q),\inf_{\alpha \in [0,1]} \left\{(\bar \epsilon +  \epsilon \alpha)D(1/2 || p)+ \epsilon I_\alpha(1/2)\right\}\right)\right]. 
\end{align*}

Finally, note that we may drop the term $D(1/2 || p)$ from the inner minimization, since when $\epsilon = 0$ (and for any choice of $\alpha$), we have
\begin{equation*}
(\bar \epsilon + \epsilon \alpha) D(1/2 \| p) + \epsilon I_\alpha(1/2) = D(1/2 \| p).
\end{equation*}
%when $\alpha = 0$ and $\epsilon \to 0$, we obtain
%\begin{align*}
%\sup_{\epsilon \in [0,1]} \left[\inf_{\alpha \in [0,1]} \left\{(\bar \epsilon +  \epsilon \alpha)D(1/2 || p)+ \epsilon I_\alpha(1/2)\right\}\right] &\leq \sup_{\epsilon \in [0,1]} \left[ \inf_{\alpha = 0} \bar \epsilon D(1/2 || p)+ \epsilon I_0(1/2) \right]\\
%
%&= \sup_{\epsilon \in [0,1]} \left[\bar \epsilon D(1/2 || p) + \epsilon D(1/2 || q)\right]\\
%
%&\leq D(1/2 || p).
%\end{align*}
Thus, the optimal learning rate is
\begin{align*}
\cR = \sup_{\epsilon \in [0,1]} \left[ \min\left( \epsilon D(1/2 || q),\inf_{\alpha \in [0,1/2]} \left\{(\bar \epsilon +  \epsilon \alpha)D(1/2 || p)+ \epsilon I_\alpha(1/2)\right\}\right)\right]. 
\end{align*}
\end{proof}

%To see the benefit of picking the optimal $\epsilon$, consider the case of a perfect student; i.e., when $q=0$. When $q=0$, the simple majority learner (corresponding to $\epsilon = 1$) learns at the rate of $D(1/2 || p)/2$. However, a better strategy is for the student to look only at her final observation. If  $B < n$, the student is assured to learn the correct $\Theta$, and thus the student's learning rate is the same as the teacher's: $D(1/2 || p)$. Clearly, this is optimal. 

Note that when $q = 0$, i.e., the student is a perfect learner, the learning rate in Theorem~\ref{ThmOptCum} corresponds to the optimum of
\begin{equation*}
\sup_{\epsilon \in [0,1]} \inf_{\alpha \in [0,1/2]} \left\{(\bar \epsilon +  \epsilon \alpha)D(1/2 || p)+ \epsilon I_\alpha(1/2)\right\},
\end{equation*}
which occurs when $\epsilon = 0$ and produces a learning rate of $D(1/2 \| p)$. In other words, the optimal strategy of the student is to estimate the state of the world based on the final transmission of the teacher, in which case her learning rate agrees with the teacher's learning rate. Although it is generally not possible to simplify the expression in Theorem~\ref{ThmOptCum} further for other choices of $q$, it is easy to calculate the learning rate to a high degree of accuracy using standard computing software. See Figure \ref{FigThree} for a comparison of the $\epsilon$-learning strategy (and the majority learning strategy) with the simple forwarding + majority learning strategy.

\begin{figure}
\begin{center}
\begin{tabular}{cc}
\includegraphics[width=3in]{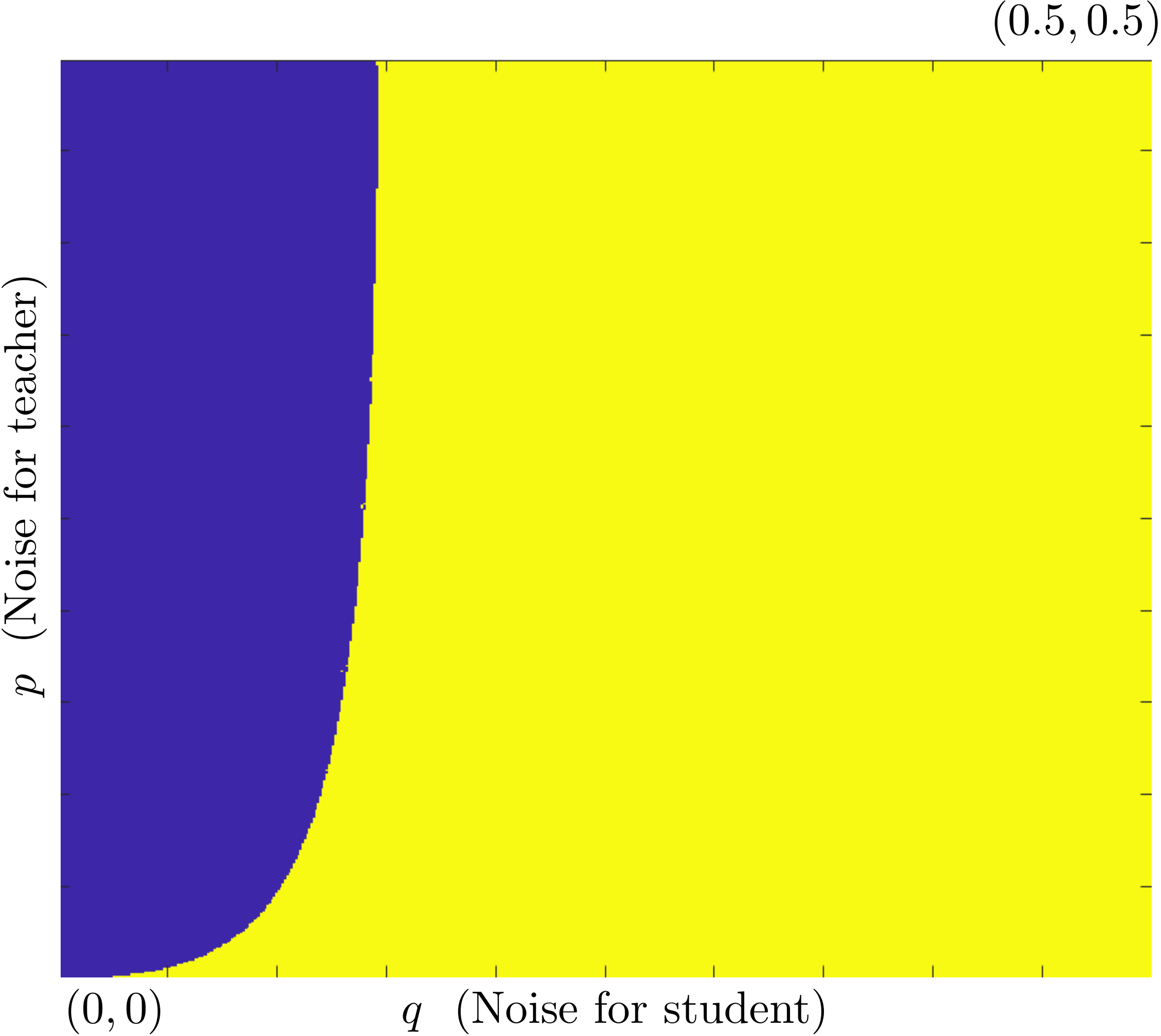} & \includegraphics[width=3in]{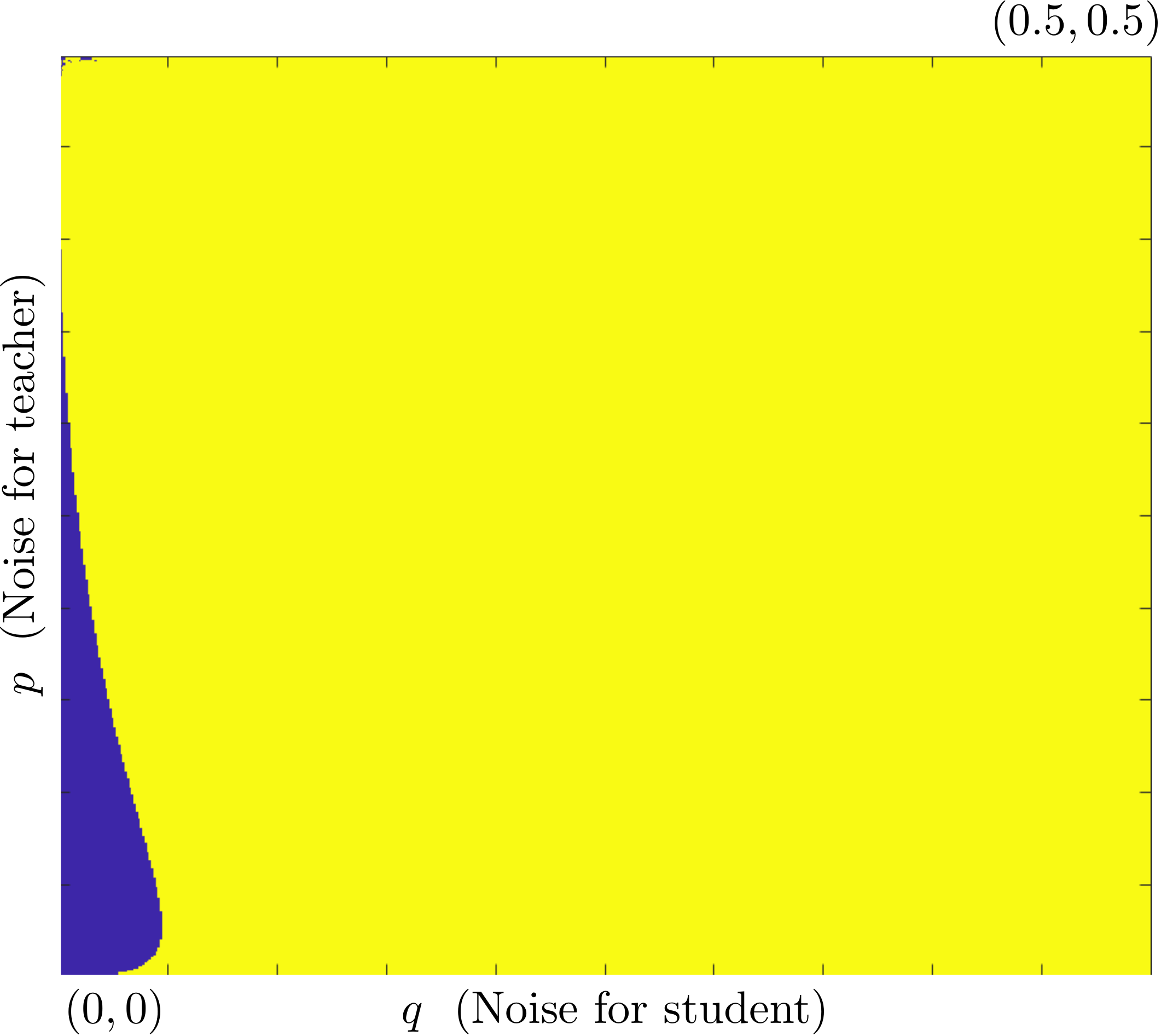} \\
(a) simple forwarding + majority learning (blue) vs.\ & (b) simple forwarding + majority learning (blue) vs.\ \\
cumulative teaching + majority learning (yellow) & cumulative teaching + $\epsilon$-majority learning (yellow)
\end{tabular}
%\caption{Plots showing the benefits of $\epsilon$-teaching + $\epsilon$-learning, analyzed in Theorem~\ref{thm: eps_teaching}. The figure on the right is simulated using the formula in Theorem \ref{ThmCricket}. In both figures, the yellow region indicates the parameter values where $\epsilon$-teaching + $\epsilon$-learning dominates the competing strategy based on majority learning. In (a), note that when the cumulative noise in the communication system is small, it is better to simply forward information. In (b), note that when the teacher's noise is large compared to the student's noise, the $\epsilon$-teaching strategy dominates the cumulative teaching strategy, since the teacher has less time to confuse himself.} 
\caption{Plots showing the benefits of $\epsilon$-majority learning over majority learning, when the teacher employs a cumulative teaching strategy. The learning rate for the simple forwarding + majority learning strategy is $D(1/2 || p \star q)$; the learning rate for the cumulative teaching + majority learning strategy is simulated using the formula in Theorem~\ref{ThmCricket}; and the learning rate for the cumulative teaching + $\epsilon$-majority learning strategy is simulated using the formula derived in Theorem~\ref{ThmOptCum}. In both figures, the blue region indicates the parameter values where simple forwarding + majority learning dominates the competing strategy. In (a), note that the simple forwarding + majority learning strategy is better for smaller values of $q$ (i.e., sharper students). However, we see a curious threshold around $q = 0.15$, such that if the student's channel is more noisy that this value, it is better to use the cumulative teaching strategy \emph{regardless} of the teacher's noise parameter. In (b), note that the blue region is significantly smaller, since the $\epsilon$-majority learning strategy is uniformly better than the majority learning strategy. As before, we observe that the simple forwarding + majority learning strategy is better for smaller values of $q$. Furthermore, for all $q$ larger than approximately 0.05, it is always preferable to use the $\epsilon$-majority learning strategy over the majority learning strategy. 
}
\label{FigThree}
\end{center}
\end{figure}

%%%%%

\section{Gaussian learning}
\label{SecGaussian}

We now turn to a continuous analog of the teacher-student learning problem. Suppose the state of the world corresponds to a parameter $\mu \in \real$. The teacher receives observations $\{y_i\}_{i=1}^n$, where $y_i \sim N(\mu, \sigma_1^2)$ are i.i.d. Furthermore, at time step $i$, the teacher transmits an estimate $x_i = f_i(y_1, \dots, y_i)$ to the student, who receives $z_i = x_i + \epsilon_i$, where $\epsilon_i \stackrel{i.i.d.}{\sim} N(0, \sigma_2^2)$. In other words, the teacher observes the state of the world through a Gaussian channel with noise variance $\sigma_1^2$, and the student observes the teacher's transmissions through a Gaussian channel with noise variance $\sigma_2^2$. The final estimate of the student after $n$ time steps is some function $\thetahat(z_1, \dots, z_n)$.

We again seek to compare different teaching/learning strategies, where the learning rate of the student is now characterized by the parameter pairs $(\sigma_1, \sigma_2)$ governing the two channels, rather than the pairs $(p,q)$. In the continuous parameter setting, we replace the notion of a majority learner (which no longer makes sense) by a learner who simply takes an average over all observations. More broadly, we allow both the teacher and student to learn by taking a linear combination of observations; i.e., the functions $\{f_i\}_{i=1}^n$ and $\thetahat$ are linear. We compare various strategies in terms of the variance $\Var(\thetahat)$  of the final estimate of the student.

Consequently, we may reparametrize the problem according to a matrix-vector pair $(A,b)$. If we use $x, y, z, \epsilon \in \real^n$ to denote the vectorized versions of the corresponding random variables, we see that the teacher receives the observation vector $y \sim N(\mu \textbf{1}, \sigma_1^2 I_n)$ and transmits the vector $Ay$. The student receives $z = Ay + \epsilon$, where $\epsilon \sim N(0, \sigma_2^2 I_n)$, and deduces $\thetahat = b^T z$.

The quality of the student's estimator may be calculated as
\begin{equation}
\label{EqnVar}
f(A,b) := \Var(b^T z) = \Var(b^T Ay + b^T \epsilon) = \sigma_1^2 b^T A A^T b + \sigma_2^2 b^T b.
\end{equation}
We seek to minimize $f$ with respect to $A$ and $b$.

Note that the fact that the teacher must transmit messages depending only on his past observations constrains the matrix $A$ to be lower-triangular. Furthermore, we impose the additional assumption that the teacher's transmissions on successive days are unbiased estimators of the state of the world, so
\begin{equation}
\label{EqnUnbiased}
\mu \textbf{1} = \E[z] = \E[Ay] = \mu A \textbf{1}.
\end{equation}
Equivalently (assuming we are in an arbitrary setting where $\mu \neq 0$), we have $A \textbf{1} = \textbf{1}$, so $A$ is a row-stochastic matrix. Similarly, we require the student to output an unbiased estimator, so
\begin{equation*}
\mu = \E[b^T z] = \E[b^T (Ay + \epsilon)] = b^T A \cdot \mu \textbf{1},
\end{equation*}
implying that $b^T A \textbf{1} = 1$, so $b^T \textbf{1} = 1$.

Due to the relatively simple form of expression~\eqref{EqnVar}, we can analyze the optimal student strategy for a fixed teacher strategy with relative ease. We can also determine a jointly optimal strategy in terms of the pair $(A,b)$.

\begin{remark*}
We briefly remark on the condition~\eqref{EqnUnbiased} that the teacher must transmit unbiased estimators of the state of the world. Note that if this were not the case, the teacher and the student could agree a priori that the teacher would simply amplify each observation by a large constant $M$, and the student would rescale the received transmissions by $\frac{1}{M}$. Thus, the student would receive the vector $z = My + \epsilon$, which would be transformed to $z' = y + \frac{1}{M} \epsilon$. As $M \rightarrow \infty$, this would correspond to a noiseless channel from the teacher to the student. By the Cramer-Rao bound, the minimum variance unbiased estimate for the teacher based on $n$ i.i.d.\ observations $\{y_i\}_{i=1}^n$ is the average, which has variance $\frac{\sigma_1^2}{n}$. Furthermore, it is possible to achieve this lower bound by taking $A = MI_n$ and $b = \frac{1}{Mn} \textbf{1}$.
\end{remark*}

\subsection{Optimal student strategy}

We first consider the optimal strategy of the student when the strategy of the teacher (corresponding to the matrix $A$) is fixed. We have the following result:

\begin{theorem}
\label{ThmGaussStudent}
Let $A$ be a fixed row-stochastic, lower-triangular matrix. The student strategy that minimizes the variance $f(A,b)$ of the estimator is given by
\begin{equation*}
b^* = \frac{(\sigma_1^2 AA^T + \sigma_2^2 I)^{-1} \textbf{1}}{\textbf{1}^T (\sigma_1^2 AA^T + \sigma_2^2 I)^{-1} \textbf{1}},
\end{equation*}
resulting in a variance of
\begin{equation*}
f(A, b^*) = \frac{1}{\textbf{1}^T (\sigma_1^2 AA^T + \sigma_2^2I)^{-1} \textbf{1}}.
\end{equation*}
\end{theorem}

\begin{proof}
The optimal strategy of the student may be obtained by optimizing the expression~\eqref{EqnVar}:
\begin{align*}
\min_b & \quad b^T(\sigma_1^2 AA^T + \sigma_2^2 I)b \\
\text{s.t.} & \quad b^T \textbf{1} = 1.
\end{align*}
We can optimize this using the method of Lagrange multipliers. Define the function
\begin{equation*}
g(b,\lambda) = b^T(\sigma_1^2 AA^T + \sigma_2^2 I) b + \lambda(b^T \textbf{1} - 1).
\end{equation*}
Then
\begin{equation*}
\frac{\partial g}{\partial b} = 2(\sigma_1^2 AA^T + \sigma_2^2 I) b + \lambda \textbf{1},
\end{equation*}
so setting $\frac{\partial g}{\partial b} = 0$ and solving for $b$ gives $b = \frac{-\lambda}{2} (\sigma_1^2 AA^T + \sigma_2^2 I)^{-1} \textbf{1}$. Hence, setting $b^T \textbf{1} - 1 = 0$ implies that
\begin{equation*}
\frac{-\lambda}{2} \textbf{1}^T (\sigma_1^2 AA^T + \sigma_2^2 I)^{-1} \textbf{1} - 1 = 0,
\end{equation*}
so
\begin{equation*}
\lambda = \frac{-2}{\textbf{1}^T (\sigma_1^2 AA^T + \sigma_2^2 I)^{-1} \textbf{1}},
\end{equation*}
implying that
\begin{equation}
\label{EqnOptStudent}
b = \frac{(\sigma_1^2 AA^T + \sigma_2^2 I)^{-1} \textbf{1}}{\textbf{1}^T (\sigma_1^2 AA^T + \sigma_2^2 I)^{-1} \textbf{1}}.
\end{equation}
The variance of the optimal strategy is then obtained by plugging the value of $b$ into the variance formula~\eqref{EqnVar}, to obtain
\begin{equation}
\label{EqnOptStudVar}
\frac{1}{\textbf{1}^T (\sigma_1^2 AA^T + \sigma_2^2I)^{-1} \textbf{1}}.
\end{equation}
\end{proof}

Note that although we focused on fixed (and relatively simplistic) learning strategies for the student in the binary case, the simpler expressions for the variance of the student's estimators allow us to derive the optimal strategy that a student should employ for a particular teaching mechanism. Indeed, different values of the teacher's matrix $A$ determine the relative weights of the vector $(\sigma_1^2 AA^T + \sigma_2^2 I)^{-1} \textbf{1}$, which govern how the student should weight her observations as time progresses. In the case $A = I$, the student's optimal strategy would be to place equal weight on all observations (i.e., simple averaging). Note that the student's optimal strategy will never correspond to a weighted average of her last $\epsilon n$ observations, since the teacher's transmissions in the first $(1-\epsilon) n$ time steps are assumed to be unbiased estimators of $\mu$, so an optimal student learner would always prefer to compute a weighted average that takes into account these initial observations, as well.

%%%%%

\subsection{Teacher strategies}

As analogs to the strategies studied in the binary teaching/learning setting, we now discuss the following classes of teacher strategies, and provide a brief comparison of the relative quality of the ensuing estimators computed by an optimal student.

\begin{enumerate}
\item \textbf{Simple forwarding:} This corresponds to $A = I_n$.
\item \textbf{$\epsilon$-teaching:} The teacher first learns for $(1-\epsilon)n$ steps, and then transmits his best estimate based on the first $(1-\epsilon)n$ observations, for the remaining $\epsilon n$ steps. The teacher's best estimate corresponds to $\frac{y_1 + \cdots + y_{\lfloor (1-\epsilon)n\rfloor}}{\lfloor (1-\epsilon)n\rfloor}$, so $A$ is a block matrix with all entries in the lower left $\epsilon n \times \lfloor(1-\epsilon) n\rfloor$ block equal to $\frac{1}{\lfloor(1-\epsilon)n\rfloor}$, and the remaining entries equal to 0.
\item \textbf{Cumulative learning:} The teacher transmits his best estimate at each step. This corresponds to $A$ being a lower-triangular matrix with all entries in column $i$ equal to $\frac{1}{i}$.
\end{enumerate}

In the simple forwarding teaching strategy, the student receives i.i.d.\ observations $z_i \sim N(\mu, \sigma_1^2 + \sigma_2^2)$, and the minimum-variance strategy is clearly for the student to take a simple average. We can also see this by plugging $A = I$ into the formula~\eqref{EqnOptStudent}:
\begin{equation*}
b = \frac{\frac{1}{\sigma_1^2 + \sigma_2^2} I \textbf{1}}{\textbf{1}^T \frac{1}{\sigma_1^2 + \sigma_2^2} I \textbf{1}} = \frac{1}{n} \textbf{1}.
\end{equation*}
The variance of the overall estimator is then equal to $\frac{\sigma_1^2 + \sigma_2^2}{n}$.

%For the other strategies, we can plug in the values for $A$ and simulate the results (\textcolor{red}{include some simulations}).

We now make two surprising observations: First, suppose the teacher employs the $\epsilon$-teaching strategy, and the student simply averages the latter $\epsilon n$ observations. As shown above, this may not agree with the optimal strategy for the student; however, this leads to closed-form expressions and a simple comparison. This is analogous to the strategy in the binary setting where the student takes a simple majority over the latter $\epsilon n$ observations. Plugging into the formula~\eqref{EqnVar}, it is not hard to see that the variance of this strategy is
\begin{equation*}
\frac{\sigma_1^2}{(1-\epsilon)n} + \frac{\sigma_2^2}{\epsilon n},
\end{equation*}
since the matrix $AA^T$ has all entries equal to 0 other than the lower right $\epsilon n \times \epsilon n$ block, which has all entries equal to $\frac{1}{(1-\epsilon) n}$. Clearly, this expression is strictly larger than $\frac{\sigma_1^2}{n} + \frac{\sigma_2^2}{n}$. Thus, contrary to the conclusions in the binary setting, the simple forwarding teacher strategy \emph{always} dominates the $\epsilon$-teaching + simple averaging strategy.

Second, suppose the teacher is clairvoyant, and is allowed to transmit estimates of the state of the world based on his entire observation vector $y$, rather than only the observations seen up to time $i$. (As an alternative interpretation, suppose the teacher first learned from $n$ observations in a previous epoch, prior to teaching.) The ``best" strategy is intuitively to set $A = \frac{1}{n} \textbf{1}\textbf{1}^T$, corresponding to repeatedly transmitting the maximum likelihood estimator for $\mu$. Based on the formula~\eqref{EqnOptStudVar}, we can see that the optimal student strategy actually produces the same variance as the best estimator in the case of simple forwarding: Note that $A$ is doubly stochastic, so $\textbf{1}$ is clearly an eigenvector of $(\sigma_1^2 AA^T + \sigma_2^2 I)$, and we can easily check that
\begin{equation*}
(\sigma_1^2 AA^T + \sigma_2^2 I) \textbf{1} = (\sigma_1^2 + \sigma_2^2) \textbf{1}.
\end{equation*}
But then
\begin{equation*}
\frac{1}{\sigma_1^2 + \sigma_2^2} \textbf{1} = (\sigma_1^2 AA^T + \sigma_2^2 I)^{-1} \textbf{1},
\end{equation*}
so we have
\begin{equation*}
\textbf{1}^T (\sigma_1^2 AA^T + \sigma_2^2 I)^{-1} \textbf{1} = \frac{n}{\sigma_1^2 + \sigma_2^2},
\end{equation*}
from which the claim follows. This is markedly different from the binary learning setting.

\subsection{Jointly optimal strategies}

The calculations from the previous subsection suggest that the simple forwarding teacher strategy may sometimes be dominant. Indeed, we now prove that this strategy is \emph{always} jointly optimal:

\begin{theorem}
The joint optimization problem
\begin{align*}
\min_{A,b} & f(A,b) \\
\text{s.t.} & \quad A \textbf{1} = \textbf{1}, \\
& \quad b^T \textbf{1} = 1,
\end{align*}
is minimized when $A^* = I_n$ and $b^* = \frac{1}{n} \textbf{1}$.
\end{theorem}

\begin{proof}
Denote the set of parameters
\begin{equation*}
\Theta := \left\{(A, b): A \textbf{1} = \textbf{1}, b^T \textbf{1} = 1\right\},
\end{equation*}
and for any $\kappa > 0$, define the set
\begin{equation*}
\Theta_\kappa := \left\{(A,b): \|b\|_2^2 \ge \kappa, b^T A \textbf{1} = 1\right\}.
\end{equation*}
Note that for $(A,b) \in \Theta$, we have $\|b\|_1 \ge b^T \textbf{1} = 1$, so $\|b\|_2^2 \ge \frac{1}{n} \|b\|_1^2 \ge \frac{1}{n}$, implying that $\Theta \subseteq \Theta_\kappa$ for $\kappa \ge \frac{1}{n}$.

Now note that
\begin{equation*}
\min_{(A,b) \in \Theta_\kappa} f(A,b) \ge \min_{(A,b) \in \Theta_\kappa} \sigma_1^2 b^T AA^T b + \sigma_2^2 \kappa.
\end{equation*}
Furthermore, if we define $w = A^T b$, we see that the expression $\sigma_1^2 b^T A A^T b$ is simply the variance of the estimator $w^T y$ of $\mu$, where the condition that $b^T A \textbf{1} = 1$ simply constrains $w^T y$ to be an unbiased estimator. By the Cramer-Rao bound, we therefore conclude that
\begin{equation*}
\min_{(A,b) \in \Theta_\kappa} f(A,b) \ge \frac{\sigma_1^2}{n} + \sigma_2^2 \kappa.
\end{equation*}
Finally, taking $\kappa = \frac{1}{n}$, we conclude that
\begin{equation}
\label{EqnCRGaussian}
\min_{(A,b) \in \Theta} f(A,b) \ge \min_{(A,b) \in \Theta_\kappa} f(A,b) \ge \frac{\sigma_1^2}{n} + \frac{\sigma_2^2}{n}.
\end{equation}
Since this lower bound is achieved when $A = I_n$ and $b = \frac{1}{n} \textbf{1}$, the simple forwarding + simple averaging strategy must always be a joint minimizer.
\end{proof}

\begin{remark*}
In fact, the preceding argument only requires $A$ to be row-stochastic, and not necessarily lower-triangular. Thus, we see that the simple forwarding teaching strategy + simple averaging learning strategy is in fact optimal even for clairvoyant teachers.
\end{remark*}

\subsection{Generalizations}

Note that the optimality argument in the previous subsection does not actually require the $\epsilon_i$'s to be Gaussian, as long as they are i.i.d.\ with variance $\sigma_2^2$. Some natural questions are whether the results also rely on Gaussianity of the teacher's observations and/or linearity of the teacher or student strategies.

Regarding the Gaussian assumption on the teacher's observations $\{y_i\}_{i=1}^n$, we note that if $\Var(z_i) = \sigma_1^2$, we similarly have the lower bound
\begin{equation*}
\min_{(A,b) \in \Theta_{\kappa}} f(A,b) \ge \min_{(A,b) \in \Theta_{\kappa}} \sigma_1^2 b^T AA^T b + \sigma_2^2 \kappa,
\end{equation*}
if the teacher and student strategies are parametrized by $A$ and $b$, respectively. Although it is no longer true that the Cramer-Rao lower bound is achieved for non-Gaussian data, the best linear unbiased estimator (BLUE) based on $n$ i.i.d.\ samples is nonetheless still achieved by the empirical average. Hence, inequality~\eqref{EqnCRGaussian} holds, with equality achieved in the case $A = I_n$, $b = \frac{1}{n} \textbf{1}$, and $\kappa = \frac{1}{n}$, as before.

Moving to the question of whether linearity of strategies is required, note that we can generalize our optimality argument in the case when the $y_i$'s are Gaussian to the case when the teacher is allowed to transmit \emph{any} family of functions $\{f_i\}_{i=1}^n$ at each time step (even functions that depend on any future observations in the set $\{y_1, \dots, y_n\}$). In such a setting, the Cramer-Rao lower bound~\eqref{EqnCRGaussian} still applies to the wider class of estimators, showing that the linear/simple forwarding strategy is optimal over the entire class. The question of whether such a statement holds when the $y_i$'s are not Gaussians remains open. We also do not currently have a characterization of the set of jointly optimal strategies when the student is allowed to employ a more sophisticated non-linear strategy. 

%%%%%

\section{Discussion and open problems}
\label{section: end}

%\begin{figure}
%\begin{center}
%\includegraphics[scale=0.3]{finalplot2}
%\caption{The dark region indicates where the low-effort strategy outperforms the high-effort strategy.}\label{fig: final}
%\end{center}
%\end{figure}

In Figure~\ref{FigThree}, we compared the learning rates for the student for various teacher strategies. Our results indicate that if the teacher to student channel has a low level of noise, it is better for the student to transmit ``uncoded" information. Intuitively, the teacher might receive many incorrect observations initially by chance, in which case the cumulative teaching strategy would have a significant delay in correcting the teacher's opinion. However, if the teacher were following the simple forwarding strategy, the flipped observations in the beginning would have no effect on the teacher's future communications. Furthermore, since the student has a relatively clean channel, she does not need a cumulative teaching strategy to learn quickly. We also noted a surprising threshold of $q \approx 0.15$ that emerged from the figure: if the teacher to student channel is more noisy than this threshold, then it is \emph{always} beneficial to use the cumulative teaching strategy, no matter how bad the teacher's channel.

Various alternative teaching and learning strategies exist that have not been analyzed here. In particular, we did not analyze Bayesian strategies for the student. Although the learning rate of such strategies could be simulated for small $n$, obtaining an accurate approximation of the (exponentially small) error probability from simulations for larger values of $n$ is challenging. We also note that analyzing specific teaching and learning strategies provide lower bounds on the best possible learning rate. 

An interesting line of inquiry is to characterize the optimal learning rate over all possible joint strategies between the teacher and student. Observe that the optimal learning rate for the teacher is $D(1/2\| p)$ and the optimal learning rate for the student, if she were to observe $\Theta$ directly through her channel, is $D(1/2 \| q)$. Simple arguments using the data processing inequality show that the optimal learning rate is at most $\min (D(1/2 \| p), D(1/2 \| q))$. Can this bound be achieved? We think this would be very surprising, leading us to formulate the following conjecture:
\begin{conj}
Let $0 < p, q < 1/2$.  We conjecture that the optimal learning rate of the student over all possible joint strategies between the teacher and the student is strictly less than $\min (D(1/2 \| p), D(1/2 \| q))$. 
\end{conj}
It is interesting to note that if the teacher is allowed to have \emph{non-causal} strategies, the upper bound of $\min (D(1/2 \| p), D(1/2 \| q))$ may be achieved: The teacher can use all $n$ time slots to learn $\Theta$, and then transmit the learned value over $n$ time slots to the student. The conjecture above essentially states that a price must be paid for using \emph{causal} teaching and learning strategies (at least in the binary case, since our analysis shows that no such penalty exists in the Gaussian setting). A harder open problem is to determine optimal joint strategies that the teacher and student could employ to maximize the student's learning rate; note that since the student will always be Bayesian in the optimal strategy, this problem boils down to identifying an optimal teaching strategy.

We also state the fascinating conjecture from Huleihel et al.~\cite{HulEtal19} alluded to in Section~\ref{section: model}, which concerns identifying the optimal learning rate:

\begin{conj}[Huleihal et al.~\cite{HulEtal19}]
Let $p = q$. In the regime $p \to 1/2$, the optimal learning rate of the student is $D(1/2 \| p) (1 + o(1))$; i.e., the optimal learning rate is the same as that of the teacher up to first-order error terms. 
\end{conj}

%We also note that our result in Section \ref{section: markov} is closely related to the famous Ballot Theorem in combinatorics \cite{AddRee08}. Theorem \ref{thm: magic} is essentially a more refined analysis of the classical Ballot Theorem setting---extending Theorem \ref{thm: magic} to more general Markov chains, such as the Brownian motion, could potentially lead to newer versions of Ballot Theorems. It would also be interesting to study strategies and rates of learning when the state of the world is non-binary, but is drawn from a larger class of possible states. Another worthwhile question to explore is how social interactions between a collection of students affects the overall learning rate of the group, where the students form an aggregate vote after receiving individual outputs from a broadcast channel.

In the case of the Gaussian learner, we showed that the teacher-student problem is of an entirely different nature, and the simplest strategy where the teacher simply forwards information and the student constructs a simple average is always optimal over the class of linear, unbiased estimators. However, we also showed that if the teacher is allowed to transmit estimators that are not unbiased, which the student subsequently decodes, the overall estimator can have an even smaller variance. In general, the question of the best estimator over different classes of teaching/learning strategies (e.g., non-linear strategies, or biased strategies with appropriate power constraints) remains open.

Finally, we concede that the teacher/student model analyzed in this paper is a vast oversimplification of reality, and the dynamics of social learning may be substantially different in practice. The simple forwarding and $\epsilon$-teaching strategies employed by the teacher have a ``learn by rote" flavor, which any experienced teacher would realize is not the best way to convey information to a student: artful teaching involves presenting concepts from  different angles, rather than simply repeating the same lesson from one day to the next. Furthermore, we have assumed that no feedback is available to the teacher from the student, whereas a teacher should be able to adapt his strategy based on how well the student is learning. A fascinating new framework in which the teacher feeds the student carefully crafted examples from a set of lessons is known as machine teaching~\cite{Zhu15}. It is also unreasonable to expect that the student (or teacher) would have knowledge of the noise parameters of the channels beforehand, and a more realistic setting might involve gradually estimating these parameters based on feedback and adapting strategies over time. After all, a student would rightfully choose to pay less attention to a teacher if she thinks he does not know what he is talking about!

%%%%%

\section*{Acknowledgments}

The authors thank the AE and anonymous reviewers for their comments and suggestions, which led to an improved manuscript. Both authors gratefully acknowledge support from NSF grant CCF-1841190, and PL acknowledges additional support from NSF grant DMS-1749857. The authors thank the organizers of the Fourteenth Annual Workshop on Probability and Combinatorics in Barbados, where extensions to the $\epsilon$-learning rate and Gaussian learning were derived.

\bibliographystyle{unsrt}
\bibliography{refs.bib}

\begin{thebibliography}{10}

\bibitem{Cha04}
C.~P. Chamley.
\newblock {\em Rational Herds: Economic Models of Social Learning}.
\newblock Cambridge University Press, 2004.

\bibitem{MosTam17}
E.~Mossel and O.~Tamuz.
\newblock Opinion exchange dynamics.
\newblock {\em Probability Surveys}, 14:155--204, 2017.

\bibitem{Viv93}
X.~Vives.
\newblock How fast do rational agents learn?
\newblock {\em The Review of Economic Studies}, pages 329--347, 1993.

\bibitem{MolEtAl13}
A.~Jadbabaie, P.~Molavi, and A.~Tahbaz-Salehi.
\newblock Information heterogeneity and the speed of learning in social
  networks.
\newblock {\em Columbia Business School Research Paper}, 2013.

\bibitem{MolEtAl17}
P.~Molavi, A.~Tahbaz-Salehi, and A.~Jadbabaie.
\newblock Foundations of non-{B}ayesian social learning.
\newblock {\em Columbia Business School Research Paper}, 2017.

\bibitem{HarEtal14}
M.~Harel, E.~Mossel, P.~Strack, and O.~Tamuz.
\newblock Rational groupthink.
\newblock {\em The Quarterly Journal of Economics}, 07 2020.
\newblock qjaa026.

\bibitem{MobRos14}
M.~Mobius and T.~Rosenblat.
\newblock Social learning in economics.
\newblock {\em Annu. Rev. Econ.}, 6(1):827--847, 2014.

\bibitem{Aum76}
R.~J. Aumann.
\newblock Agreeing to disagree.
\newblock {\em The Annals of Statistics}, pages 1236--1239, 1976.

\bibitem{GeaPol82}
J.~D. Geanakoplos and H.~M. Polemarchakis.
\newblock We can't disagree forever.
\newblock {\em Journal of Economic Theory}, 28(1):192--200, 1982.

\bibitem{Ban92}
A.~V. Banerjee.
\newblock A simple model of herd behavior.
\newblock {\em The Quarterly Journal of Economics}, 107(3):797--817, 1992.

\bibitem{BikEtal92}
S.~Bikhchandani, D.~Hirshleifer, and I.~Welch.
\newblock A theory of fads, fashion, custom, and cultural change as
  informational cascades.
\newblock {\em Journal of Political Economy}, 100(5):992--1026, 1992.

\bibitem{SmiSor00}
L.~Smith and P.~S{\o}rensen.
\newblock Pathological outcomes of observational learning.
\newblock {\em Econometrica}, 68(2):371--398, 2000.

\bibitem{GalKar03}
D.~Gale and S.~Kariv.
\newblock Bayesian learning in social networks.
\newblock {\em Games and Economic Behavior}, 45(2):329--346, 2003.

\bibitem{MosEtal14}
E.~Mossel, A.~Sly, and O.~Tamuz.
\newblock Asymptotic learning on {B}ayesian social networks.
\newblock {\em Probability Theory and Related Fields}, 158(1-2):127--157, 2014.

\bibitem{Deg74}
M.~H. DeGroot.
\newblock Reaching a consensus.
\newblock {\em Journal of the American Statistical Association},
  69(345):118--121, 1974.

\bibitem{EllFud93}
G.~Ellison and D.~Fudenberg.
\newblock Rules of thumb for social learning.
\newblock {\em Journal of Political Economy}, 101(4):612--643, 1993.

\bibitem{BalGoy98}
V.~Bala and S.~Goyal.
\newblock Learning from neighbours.
\newblock {\em The Review of Economic Studies}, 65(3):595--621, 1998.

\bibitem{RahJad16}
M.~A. Rahimian and A.~Jadbabaie.
\newblock Bayesian heuristics for group decisions.
\newblock {\em arXiv preprint arXiv:1611.01006}, 2016.

\bibitem{JadEtal17}
J.~Haz{\l}a, A.~Jadbabaie, E.~Mossel, and M.~A. Rahimian.
\newblock Bayesian decision making in groups is hard.
\newblock {\em Operations Research}.
\newblock To appear.

\bibitem{KanTam13}
Y.~Kanoria and O.~Tamuz.
\newblock Tractable {B}ayesian social learning on trees.
\newblock {\em IEEE Journal on Selected Areas in Communications},
  31(4):756--765, 2013.

\bibitem{Ho80}
Y-C. Ho.
\newblock Team decision theory and information structures.
\newblock {\em Proceedings of the IEEE}, 68(6):644--654, 1980.

\bibitem{Wit68}
H.~S. Witsenhausen.
\newblock A counterexample in stochastic optimum control.
\newblock {\em SIAM Journal on Control}, 6(1):131--147, 1968.

\bibitem{And74}
E.~S. Andersen.
\newblock Survey of fluctuation theory.
\newblock {\em Advances in Applied Probability}, 6(2):213--214, 1974.

\bibitem{AddRee08}
L.~Addario-Berry and B.~A. Reed.
\newblock Ballot theorems, old and new.
\newblock In {\em Horizons of Combinatorics}, pages 9--35. Springer, 2008.

\bibitem{ChuFel49}
K.~L. Chung and W.~Feller.
\newblock On fluctuations in coin-tossing.
\newblock {\em Proceedings of the National Academy of Sciences},
  35(10):605--608, 1949.

\bibitem{And55}
E.~S. Andersen.
\newblock On the fluctuations of sums of random variables ii.
\newblock {\em Mathematica Scandinavica}, pages 195--223, 1955.

\bibitem{Lev40}
P.~L\'{e}vy.
\newblock Sur certains processus stochastiques homog\`{e}nes.
\newblock {\em Compositio Mathematica}, 7:283--339, 1940.

\bibitem{Dur19}
R.~Durrett.
\newblock {\em Probability: Theory and Examples}, volume~49.
\newblock Cambridge University Press, 2019.

\bibitem{HulEtal19}
W.~Huleihel, Y.~Polyanskiy, and O.~Shayevitz.
\newblock Relaying one bit across a tandem of binary-symmetric channels.
\newblock In {\em 2019 IEEE International Symposium on Information Theory
  (ISIT)}, pages 2928--2932. IEEE, 2019.

\bibitem{Sta15}
R.~P. Stanley.
\newblock {\em Catalan Numbers}.
\newblock Cambridge University Press, 2015.

\bibitem{DupEll11}
P.~Dupuis and R.~S. Ellis.
\newblock {\em A Weak Convergence Approach to the Theory of Large Deviations},
  volume 902.
\newblock John Wiley \& Sons, 2011.

\bibitem{Zhu15}
X.~Zhu.
\newblock Machine teaching: {A}n inverse problem to machine learning and an
  approach toward optimal education.
\newblock In {\em Twenty-Ninth AAAI Conference on Artificial Intelligence},
  2015.

\end{thebibliography}

\newpage

\begin{appendix}
\section{Proof of Lemma \ref{lemma: gamma}}\label{proof: lemma: gamma}
It is enough to show that $\E e^{\lambda_1 \tilde T + \lambda_2 |\tilde T|} \leq \frac{1}{2p}$. Using the generating function~\eqref{EqnGenFunc}, we may calculate
\begin{align*}
\mathbb Ee^{\lambda_1 \tilde T + \lambda_2 |\tilde T|} &= \sum_{k=1}^\infty P(\tilde T = 2k) \Big( e^{(\lambda_1 + \lambda_2) 2k} + e^{(\lambda_2-\lambda_1)2k}\Big)\\
&= \sum_{k=1}^\infty \frac{1}{2p}(p\bar p)^k C_{k-1} \Big( e^{(\lambda_1 + \lambda_2) 2k} + e^{(\lambda_2-\lambda_1)2k}\Big)\\ 
&= \frac{\bar p}{2} \left(e^{2(\lambda_1+\lambda_2)} \sum_{k=0}^\infty C_{k} \left(p\bar p e^{2(\lambda_1+\lambda_2)}\right)^k+ e^{2(\lambda_2-\lambda_1)}\sum_{k=0}^\infty C_{k} \left(p\bar p e^{2(\lambda_2-\lambda_1)}\right)^k \right)\\
&=  \frac{\bar p}{2} \left(e^{2(\lambda_1+\lambda_2)} f\left(p\bar p e^{2(\lambda_1+\lambda_2)}\right) + e^{2(\lambda_2-\lambda_1)} f\left(p\bar p e^{2(\lambda_2-\lambda_1)}\right) \right)\\
&= \frac{2 - \sqrt{1-4p\bar p e^{2(\lambda_1+\lambda_2)}}- \sqrt{1- 4p\bar p e^{2(\lambda_2-\lambda_1)}}}{4p}.
\end{align*}

Suppose $\lambda_1 + \lambda_2 > D(1/2 \| p)$. Then
\begin{equation*}
\lambda_1 + \lambda_2 \ge \frac{1}{2} \log \left(\frac{1}{2p}\right) + \frac{1}{2} \log\left(\frac{1}{2\bar p}\right) = \frac{1}{2} \log\left(\frac{1}{4p\bar p}\right).
\end{equation*}
Thus,
\begin{equation*}
4p\bar p e^{2(\lambda_1 + \lambda_2)} > 1,
\end{equation*}
so the moment generating function is undefined, and $L(\lambda_1, \lambda_2) = +\infty$. We can argue similarly if $- \lambda_1 + \lambda_2 > D(1/2 \| p)$.

Now consider $(\lambda_1, \lambda_2) \in \cD.$ We have
\begin{align*}
\lambda_1+ \lambda_2 & \le D(1/2 || p), \\
-\lambda_1+ \lambda_2 & \le D(1/2 || p).
\end{align*}
It is easy to see that for fixed $p$, the maximum value of the moment generating function expression is attained when $\lambda_1=0$ and $\lambda_2 = D(1/2 || p)$, and that this value is $1/2p$. This concludes the proof.

%%%%%

\section{Proof of Lemma~\ref{LemMnSecond}}
\label{AppLemMnSecond}

We may rewrite the probability as $\mathbb P(B > n) \mathbb P(M_n \leq n(1-\delta) | B > n)$.
%\begin{align*}
%\mathbb P(B > n, M_n \leq n(1-\delta)) &= \mathbb P(B > n) \mathbb P(M_n \leq n(1-\delta) | B > n).
%\end{align*}
Furthermore, conditioned on the event $\{B > n\}$, the random variable $\frac{M_n}{n}$ has a symmetric distribution around $1/2$, since the mirror image of every path up to time $n$ has the exact same probability as the original path when conditioned on the event $\{B > n\}$. Thus,
\begin{equation*}
\frac{1}{2} \leq  \mathbb P\left(\frac{M_n}{n} \leq (1-\delta) \Big | B > n\right) \leq 1,
\end{equation*}
implying that
\begin{equation*}
\mathbb P (M_n \le n (1-\delta), B>n) = \Theta\Big (\mathbb P(B > n)\Big).
\end{equation*}

We now upper-bound
\begin{align*}
\mathbb P(B > n) & = \sum_{i =  n/2}^\infty \mathbb P(B = 2i) = \sum_{i = n/2}^\infty {2i \choose i} p^i (1-p)^{i} (1-2p) \\
& \le \sum_{i = n/2}^\infty 2^{2i} p^i (1-p)^{i} (1-2p)\\
&= (4p\bar p)^{n/2} \sum_{i=0}^\infty (4p\bar p)^i (1-2p)\\
&= (4p\bar p)^{n/2} \cdot \frac{1-2p}{1-4p\bar p}\\
&= (4p\bar p)^{n/2} \cdot \frac{1}{1-2p}\\
&= e^{-n(D(1/2 || p))} \cdot \frac{1}{1-2p}.
\end{align*}
Furthermore, letting $m = \lceil \frac{n+1}{2} \rceil$, we have
\begin{align*}
\mathbb P(B > n) & \ge \mathbb P(B = 2m) = {2m \choose m} p^m (1-p)^m (1-2p) \\
& \ge \frac{2^{2m}}{2m} p^m (1-p)^m (1-2p) \\
& \ge \frac{4^{n/2}}{n+1} (p \bar p)^{n/2} p \bar p (1-2p) \\
& = e^{-n(D(1/2 \| p))} \cdot \frac{p \bar p (1-2p)}{n+1}.
\end{align*}

Combining the upper and lower bounds, we clearly have $\mprob(B > n) = e^{-n(D(1/2 || p) + o(1))}$, as claimed.

%%%%%

\section{Proof of Lemma \ref{lemma: lobo}}
\label{proof: lemma: lobo}

Consider the following event: $\left\{G= 1 \text{ and } \tilde T_1 < -\lceil n\delta \rceil\right\}$. This corresponds to the event that there is only one return to 0, but sojourn time is at least $n\delta$ and the sojourn is on the negative side of the integers. If this event occurs, then $M_n$ can at most be $n(1-\delta)$, giving us the lower bound
\begin{align*}
\mathbb P(M_n \leq n(1-\delta)) &\geq \mathbb P(G=1, \tilde T_1 \leq \lfloor -n \delta \rfloor)\\
&\geq \mathbb P(G=1, \tilde T_1 =  -n \tilde \delta_n ),
\end{align*}
where $\tilde \delta_n$ is such that $-n \tilde \delta_n$ is an even integer that is at most $ \lfloor -n \delta \rfloor$. Clearly, we can take $\delta_n \to \delta$ as $n \to \infty$. Continuing, we have
\begin{align*}
\mathbb P(G=1, \tilde T_1 =  -n \tilde \delta_n ) &= (1-2p) \times \frac{1}{n \tilde \delta_n/2 + 1}{n \tilde \delta_n - 2 \choose n\tilde \delta_n/2 - 1} (p\bar p)^{n \tilde \delta_n/2}.
\end{align*}
Taking logarithms, dividing by $n$, and taking the $\liminf$ as $n$ tends to infinity, we obtain
\begin{align*}
\liminf_{n \to \infty} \frac{1}{n} \log \P(Z_n \in \cQ) &\geq \liminf_{n\to \infty} \frac{1}{n} \log \left( (1-2p) \times \frac{1}{n\tilde \delta_n/2 + 1}{n\tilde \delta_n - 2 \choose n\tilde \delta_n/2 - 1} (p\bar p)^{n\tilde \delta_n/2}\right)\\
&= \liminf_{n \to \infty} \frac{1}{n} \log {n\tilde \delta_n - 2 \choose n \tilde \delta_n/2 - 1} (p\bar p)^{n\tilde \delta_n/2}\\
&= \liminf_{n \to \infty} \left(\tilde \delta_n \log 2 + \frac{\tilde \delta_n}{2} \log (p\bar p)\right)\\
&= -\delta \cdot \frac{1}{2} \log \frac{1}{4p\bar p}\\
&= -\delta D(1/2 || p).
\end{align*}

%%%%%

\section{Proof of Lemma \ref{lemma: cricket}}
\label{proof: lemma: cricket}

Using the independence of the $U_i$'s and $V_j$'s, we may compute the limit
\begin{align*}
\lim_{n \to \infty} \frac{1}{n} \log \E e^{n\lambda W_n} &= \lim_{n \to \infty} \frac{1}{n} \log \E e^{\lambda \left(\sum_{i=1}^{n-\lfloor n\theta \rfloor} U_i +  \sum_{j=1}^{\lfloor n\theta \rfloor} V_j\right)}\\
&= \lim_{n \to \infty} \frac{1}{n} \left((n-\lfloor n\theta \rfloor) \log (e^\lambda \bar q+q)  + \lfloor n\theta \rfloor \log (e^\lambda q + \bar q)\right)\\
&= \bar \theta \log (e^\lambda \bar q+q) +  \theta \log (e^\lambda q + \bar q).
\end{align*}
Thus, we may use the G\"{a}rtner-Ellis theorem to conclude that $W_n$ satisfies the large deviation principle with rate function
\begin{align*}
I_\theta(w) &= \left(\bar \theta \log (e^\lambda \bar q+q) +  \theta \log (e^\lambda q + \bar q)\right)^*(w)\\
&= \sup_{\lambda}\left\{ \lambda w - \bar \theta \log (e^\lambda \bar q+q) -  \theta \log (e^\lambda q + \bar q)\right\}.
\end{align*}
Differentiating with respect to $\lambda$, we see that the above supremum is attained when the following equality is satisfied:
\begin{align}\label{eq: z to lambda}
w = \frac{\bar \theta e^\lambda \bar q}{e^\lambda \bar q+q} + \frac{ \theta e^\lambda q}{e^\lambda q + \bar q}.
\end{align}
This is a quadratic equation in $e^\lambda$, which we may solve to obtain
\begin{align*}
e^\lambda &= \frac{-\tau(\theta,w) + \sqrt{\tau(\theta,w)^2+4w\bar w}}{2\bar w},
\end{align*}
where 
\begin{align*}
\tau(\theta,w) \defn  \frac{\bar q}{q}(\bar \theta - w) + \frac{q}{\bar q}(\theta - w).
\end{align*}
The rate function is then given by
\begin{align*}
I_\theta(w) &= w\log \left(  \frac{-\tau(\theta,w) + \sqrt{\tau(\theta,w)^2+4w\bar w}}{2\bar w}\right)  - \bar \theta \log \left( \bar q \left(\frac{-\tau(\theta,w) + \sqrt{\tau(\theta,w)^2+4w\bar w}}{2\bar w}\right) + q\right)\\
& \qquad - \theta \log \left(q\left(\frac{-\tau(\theta,w) + \sqrt{\tau(\theta,w)^2+4w\bar w}}{2\bar w}\right) + \bar q\right).
\end{align*}

As a sanity check, when $\theta = 0$, we have $e^\lambda = \frac{qw}{\bar q\bar w}$, and the rate function is
\begin{align*}
I_\theta(w) &= w \log \frac{qw}{\bar q \bar w} - \log \left(\frac{qw}{\bar w} + q\right)\\
&= w \log \frac{qw}{\bar q \bar w} - \log \left(\frac{q}{\bar w}\right)\\
&= D(w || \bar q),
\end{align*}
which is what we expect. We also note that when $w = \bar q\bar \theta + q\theta$, the solution to equation~\eqref{eq: z to lambda} is $e^\lambda = 1$, which gives $I_\theta(w) = 0$.

\end{appendix}
\end{document}